\documentclass[journal,onecolumn,draftcls]{IEEEtran}

\newtheorem{thm}{Theorem}
\newtheorem{cor}{Corollary}
\newtheorem{lem}{Lemma}

\newcommand{\entropy}[1]{\mathsf{h}\left(#1\right)}

\usepackage[final]{graphicx}
\usepackage[reqno]{amsmath}
\usepackage{amssymb}
\usepackage{subfig}
\usepackage{epstopdf}
\usepackage{bm}

\begin{document}

\title{Interference Channels with Coordinated Multi-Point Transmission: Degrees of Freedom, Message Assignment, and Fractional Reuse}
\author{{\large{Aly El Gamal, {\em Student Member, IEEE}, V. Sreekanth Annapureddy, {\em Student Member, IEEE}, and Venugopal V. Veeravalli, {\em Fellow, IEEE}}}
\thanks{The authors are with the Coordinated Science Laboratory and the
Department of Electrical and Computer Engineering,
University of Illinois at Urbana-Champaign, Urbana, IL 61801 USA (e-mail: elgamal1@illinois.com, sreekanthav@gmail.com, vvv@illinois.edu).}
\thanks{This paper was presented in part at the $46^{th}$ Annual Conference on Information Sciences and Systems (CISS), Princeton, NJ, Mar. 2012, and in part at the IEEE International Conference on Communications (ICC), Ottawa, ON, Jun. 2012.}
\thanks{This research was supported in part by the NSF award CCF-0904619, through the University of Illinois at Urbana-Champaign, and grants from Intel and Motorola Solutions. }
}
\maketitle

\begin{abstract}
Coordinated Multi-Point (CoMP) transmission is an infrastructural enhancement under consideration for next generation wireless networks. In this work, the capacity gain achieved through CoMP transmission is studied in various models of wireless networks that have practical significance. The capacity gain is analyzed through the degrees of freedom (DoF) criterion. The DoF available for communication provides an analytically tractable way to characterize the capacity of interference channels. The considered channel model has $K$ transmitter/receiver pairs, and each receiver is interested in one unique message from a set of $K$ independent messages. Each message can be available at more than one transmitter. The maximum number of transmitters at which each message can be available, is defined as the \emph{cooperation order} $M$. For fully connected interference channels, it is shown that the asymptotic per user DoF, as $K$ goes to infinity, remains at $\frac{1}{2}$ as $M$ is increased from $1$ to $2$. Furthermore, the same negative result is shown to hold for all $M \geq 2$ for any message assignment that satisfies a \emph{local cooperation} constraint. On the other hand, when the assumption of full connectivity is relaxed to \emph{local connectivity}, and each transmitter is connected only to its own receiver as well as $L$ neighboring receivers, it is shown that local cooperation is optimal. The asymptotic per user DoF is shown to be at least $\max\left\{\frac{1}{2},\frac{2M}{2M+L}\right\}$ for locally connected channels, and is shown to be $\frac{2M}{2M+1}$ for the special case of \emph{Wyner's asymmetric model} where $L=1$. An interesting feature of the proposed achievability scheme is that it relies on simple zero-forcing transmit beams and does not require symbol extensions. Also, to achieve the optimal per user DoF for Wyner's model, messages are assigned to transmitters in an asymmetric fashion unlike traditional assignments where message $i$ has to be available at transmitter $i$. It is also worth noting that some receivers have to be inactive, and \emph{fractional reuse} is needed to achieve equal DoF for all users.
%Combinatorial upper bound for fully connected model
%M=2 result for fully connected model
%Local Cooperation for FC
%Local Cooepration optimal for Locally connected model
%Result for locally connected model
%Message assignment and Fractional Reuse
\end{abstract}

\section{Introduction}\label{sec:introduction}

%wireless data demands

In the past decade, there has been a significant growth in the usage of wireless networks, and in particular, cellular networks, because of the increased data demands. This has been the driver of recent research for new ways of managing interference in wireless networks.

%why interference management

Due to the superposition and broadcast properties of the wireless medium, interfering signals pose a significant limitation to the rate of communication of users in a wireless network. Hence, it is of interest to understand the fundamental limits of communication in interference channels and to capture the effect of interference on optimal encoding and decoding schemes. The problem of finding the capacity region of even the simple $2-$user Gaussian interference channel is still an open problem. However, approximations exist in the literature, where the capacity region or the sum capacity is known in the special scenarios of strong and low interference (\cite{Carleial-IT77},~\cite{Shang-Kramer-Chen-IT09},~\cite{Motahari-Khandani-IT09},~\cite{Annapureddy-Veeravalli-IT09}). Another effective approximation that simplifies the problem is to consider only the sum degrees of freedom (DoF) or the pre-log factor of the sum capacity at high signal-to-noise ratio (SNR). The DoF criterion provides an analytically tractable way to characterize the sum capacity and captures the number of interference-free sessions that can be supported in a given multi-user channel. 

%why infrastructural enhancements like CoMP

In~\cite{Madsen-Nosratinia}, the DoF per user of the fully connected Gaussian interference channel was shown to be upper bounded by $1/2$. This was shown to be achievable through the interference alignment (IA) scheme in~\cite{Cadambe-IA}. However, this achievable DoF many not be sufficient to meet the demands of wireless applications in many scenarios of practical interest, and hence, it is of interest to study ways to enhance the infrastructure of wireless networks in order to increase the rate of communication.

%Fully connected IC with CoMP

The considered infrastructural enhancement in this work is the deployment of a backhaul link, through which the transmitters/base stations can exchange messages that they wish to deliver in a cellular downlink session\footnote{The considered scenario has more practical relevance than the the cellular uplink model where base station receivers can cooperate by sharing analog signals. Nevertheless, as discussed in~\cite{Gomadam-Cadambe-Jafar-IT11} for fully connected channels, the results obtained using linear schemes in our transmitter cooperation model can be obtained in the dual receiver cooperation model.}. In order to model the finite capacity of the backhaul link, we impose a cooperation constraint where each message can be available at a maximum of $M$ transmitters. We call $M$ the \emph{cooperation order}. The availability of each message at more than one transmitter allows for Coordinated Multi-Point (CoMP) transmission~\cite{CoMP-book}. In~\cite{Annapureddy-ElGamal-Veeravalli-IT11}, a CoMP transmission model for the fully connected $K$-user interference channel was considered. Each message was assumed to be available at the transmitter carrying the same index as the message as well as $M-1$ succeeding transmitters. Using an extension of the asymptotic interference alignment scheme of~\cite{Cadambe-IA}, the DoF of the channel in this setting was shown to be lower bounded by $\frac{K+M-1}{2}, \forall K < 10$, and it was conjectured that this lower bound is valid for all values of $K$. It was then shown that this lower bound is within one degree of freedom of the maximum achievable DoF. We note that this DoF cooperation gain beyond $K/2$ does not scale linearly with $K$ as $K$ goes to infinity. In other words, the asymptotic per user DoF remains $1/2$. In Section~\ref{sec:fc}, we study whether there exists an assignment of messages satisfying the cooperation order constraint that enables the achievability of an asymptotic per user DoF that is strictly greater than $1/2$.

%Why local connectivity

The assumption of full connectivity is key to the results obtained in~\cite{Madsen-Nosratinia},~\cite{Cadambe-IA},~\cite{Annapureddy-ElGamal-Veeravalli-IT11}, and in Section~\ref{sec:fc} of this work. For the fully connected interference channel, interference mitigating schemes are designed to avoid the interference caused by all other transmitters in the network. However, in practice, each receiver gets most of the destructive interference from a few dominant interfering transmitters. For example, in cellular networks, the number of dominant interfering transmit signals at each receiver ranges from two to seven. All the interference from the remaining transmitters may contribute to the interference floor, and the improvement obtained by including them in the dominant interferers set may not justify the corresponding overhead.  For this reason, we study locally connected channels in Section~\ref{sec:lc}, where the channel coefficients between transmitters and receivers that lie at a distance that is greater than some threshold are approximated to equal zero.

%Locally connected IC with CoMP

For the locally connected channel model, we assume that each transmitter is connected to $L$ neighboring receivers as well as the receiver carrying its own index, $\left \lfloor \frac{L}{2} \right \rfloor$ preceding receivers and $\left \lceil \frac{L}{2} \right \rceil$ succeeding receivers. The special case of this model where $L=1$ is Wyner's asymmetric model~\cite{Wyner}. This special case was considered in~\cite{Lapidoth-Shamai-Wigger-ISIT07}, and it was assumed that each message is available at the transmitter carrying the same index as well as $M-1$ succeeding transmitters. The asymptotic per user DoF was characterized under this setting as $\frac{M}{M+1}$. The achieving scheme relies only on zero-forcing transmit beamforming. In Section~\ref{sec:lc}, we extend this result and characterize the asymptotic per user DoF for Wyner's asymmetric model as $\frac{2M}{2M+1}$ under a general cooperation order constraint. The message assignment enabling this result uses only local cooperation, that is, each message is available only at neighboring transmitters. The size of the neighborhood does not scale linearly with the size of the network, and therefore, our assignment scheme enjoys the same advantage as the message assignment considered in~\cite{Lapidoth-Shamai-Wigger-ISIT07}. 

\subsection{Document Organization}
The remainder of this work is organized as follows. Related work is summarized in Section~\ref{sec:relatedwork}. The problem setup is then introduced in Section~\ref{sec:problemsetup}. An informal summary of results is provided in Section~\ref{sec:informalsummary}. The asymptotic per user DoF of the fully connected interference channel with CoMP transmission is studied in Section~\ref{sec:fc}. The locally connected channel model is considered in Section~\ref{sec:lc}. The introduced results are then discussed in Section~\ref{sec:discussion}, and the paper is concluded with some final remarks in Section~\ref{sec:conclusions}.

\subsection{Related Work}\label{sec:relatedwork}

Many existing works studying interference networks with cooperating transmitters use the term \emph{cognitive radios} (e.g.~\cite{Devroye-Mitran-Tarokh-IT06},~\cite{Lapidoth-Shamai-Wigger-ITW07},~\cite{Jovicic-Viswanath-IT09},~\cite{Lapidoth-Levy-Shamai-Wigger-ISIT09}, ~\cite{Lapidoth-Levy-Shamai-Wigger-arXiv12}). Cooperation through cumulative message sharing is studied for the fully connected channel in~\cite{Maamari-Tuninetti-Devroye-arXiv13}, where each message is available at the transmitter carrying the same index and all following transmitters. We use a similar setting of cooperation to that of cumulative message sharing in the coding scheme for locally connected channels in Section~\ref{sec:dofgains}. In another body of work, unlike the considered setting where we assume that transmitters cooperate by sharing complete messages, cooperation through sharing partial message information that is considered as side information is studied (see e.g.,~\cite{Devroye-Sharif-ISIT07}). In~\cite{Wang-Tse-ISIT10} and~\cite{Prabhakaran-Viswanath-IT11}, the transmitters are allowed to cooperate through noise free bit pipes or over the air, respectively. 

Communication scenarios with cooperating multiple antenna transmitters have been considered in~\cite{Ali-Motahari-Khandani-IT08} and~\cite{Jafar-Shamai-IT08} under the umbrella of the x-channel. However, in the x-channel, mutually exclusive parts of each message are given to different transmitters. This is extended in~\cite{Annapureddy-ElGamal-Veeravalli-ISIT11} to allow each part of each message to be available at more than one transmitter, and in~\cite{Devroye-Sharif-ISIT07} the MIMO x-channel is studied in the setting where transmitters share further side information.

Finally, it is worth noting that in the considered setting, we implicitly assume the coordinated design of the transmit beams between all transmitters. This kind of coordination is also referred to in the literature as \emph{transmitter cooperation}, even without the sharing of messages (see e.g.~\cite{Zhang-Cui-SP10}).

\section{Problem Setup}\label{sec:problemsetup}

We use the standard model for the $K-$user interference channel with single-antenna transmitters and receivers,
\begin{equation}
Y_i(t) = \sum_{j=1}^{K} H_{ij}(t) X_j(t) + Z_i(t),
\end{equation}
where $t$ is the time index, $X_j(t)$ is the transmitted signal of transmitter $j$, $Y_i(t)$ is the received signal at receiver $i$, $Z_i(t)$ is the zero mean unit variance Gaussian noise at receiver $i$, and $H_{ij}(t)$ is the channel coefficient from transmitter $j$ to receiver $i$ over the time slot $t$. We remove the time index in the rest of the paper for brevity unless it is needed. 

We use $[K]$ to denote the set $\{1,2,\ldots,K\}$. For any set ${\cal A} \subseteq [K]$, we define the complement set $\bar{\cal A} = \{i: i\in[K], i\notin {\cal A}\}$. For each $i \in [K]$, let $W_i$ be the message intended for receiver $i$. We use the abbreviations $W_{\cal A}$, $X_{\cal A}$, $Y_{\cal A}$, and $Z_{\cal A}$ to denote the sets $\{W_i, i\in {\cal A}\}$, $\left\{X_i, i\in {\cal A}\right\}$, $\left\{Y_i, i\in {\cal A}\right\}$, and $\left\{Z_i, i\in {\cal A}\right\}$, respectively, and the abbreviations $X_{\cal A}^n$, $Y_{\cal A}^n$, and $Z_{\cal A}^n$ to denote the sets $\left\{X_i(t), i\in {\cal A},t\in[n]\right\}$, $\left\{Y_i(t), i\in {\cal A},t\in[n]\right\}$, and $\left\{Z_i(t), i\in {\cal A},t\in[n]\right\}$, respectively. Finally, for ${\cal A},{\cal B}\subseteq[K]$, we let ${\bm H}_{{\cal A},{\cal B}}$ be the $|{\cal A}| \times |{\cal B}|$ matrix of channel coefficients between $X_{\cal B}$ and $Y_{\cal A}$.

\subsection{Channel Model}\label{sec:channelmodel}

We consider two different channel models in the sequel. First, we consider a fully connected interference channel where all channel coefficients are drawn independently from a continuous distribution. We next consider a locally connected channel model where channel coefficients between well separated nodes are approximated to be identically zero. The locally connected channel model is a function of the number of interferers $L$ as follows:

\begin{equation}\label{eq:channel}
H_{ij} \text{ is not identically } 0 \text { if and only if } j \in \left[i- \left \lceil \frac{L}{2} \right \rceil , i+ \left \lfloor \frac{L}{2} \right \rfloor \right],
\end{equation}
and all channel coefficients that are not identically zero are drawn independently from a continuous distribution.
We note that for values of $L=1$ and $L=2$, the locally connected channel reduces to the commonly known Wyner's asymmetric and symmetric linear models, respectively~\cite{Wyner}. We illustrate examples for the described fully and locally connected channel models in Figure~\ref{fig:channelmodel}.

\begin{figure}
  \centering
\subfloat[]{\label{fig:fc}\includegraphics[height=0.25\textwidth]{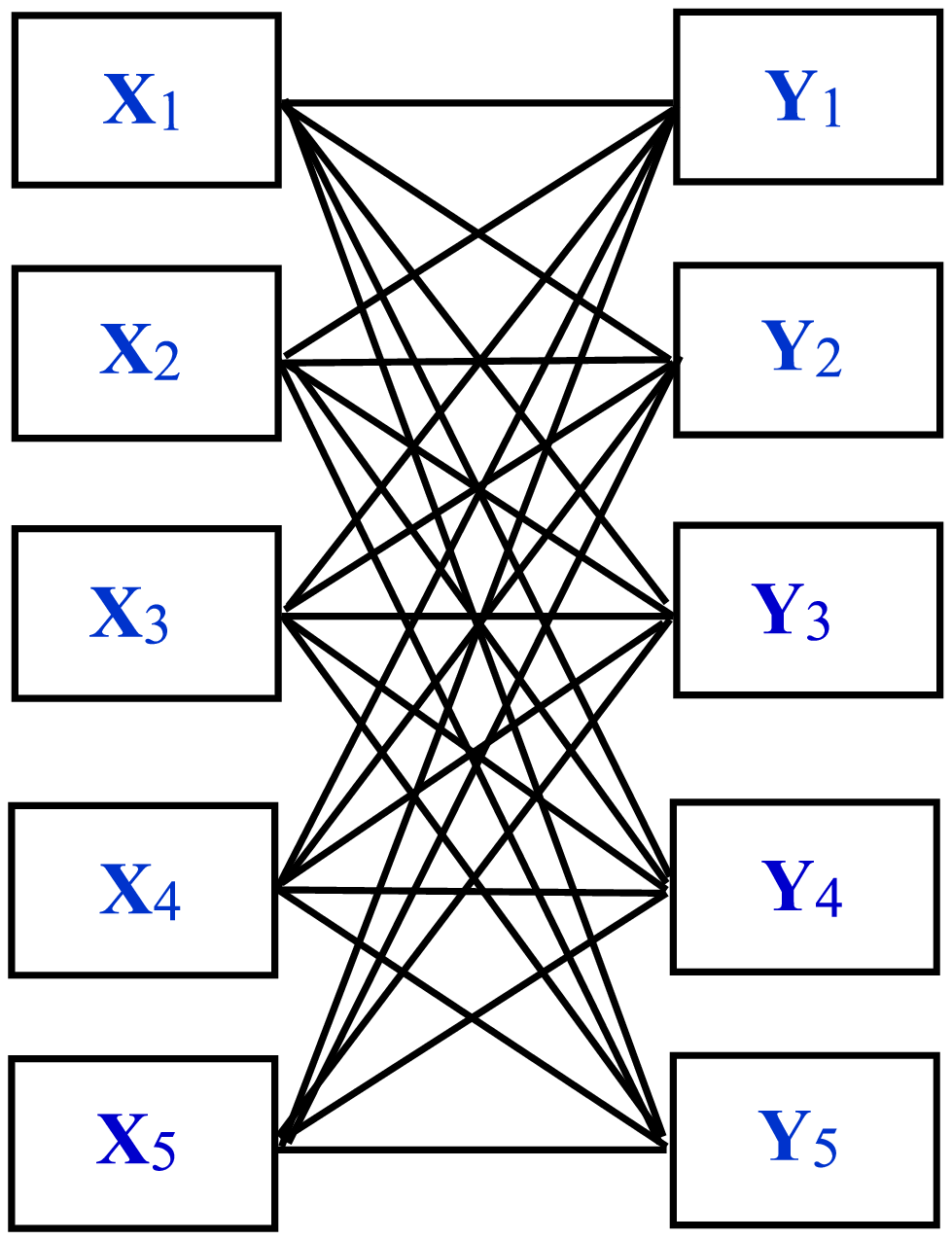}}                
\quad\quad\quad\quad\subfloat[]{\label{fig:lc}\includegraphics[width=0.188\textwidth]{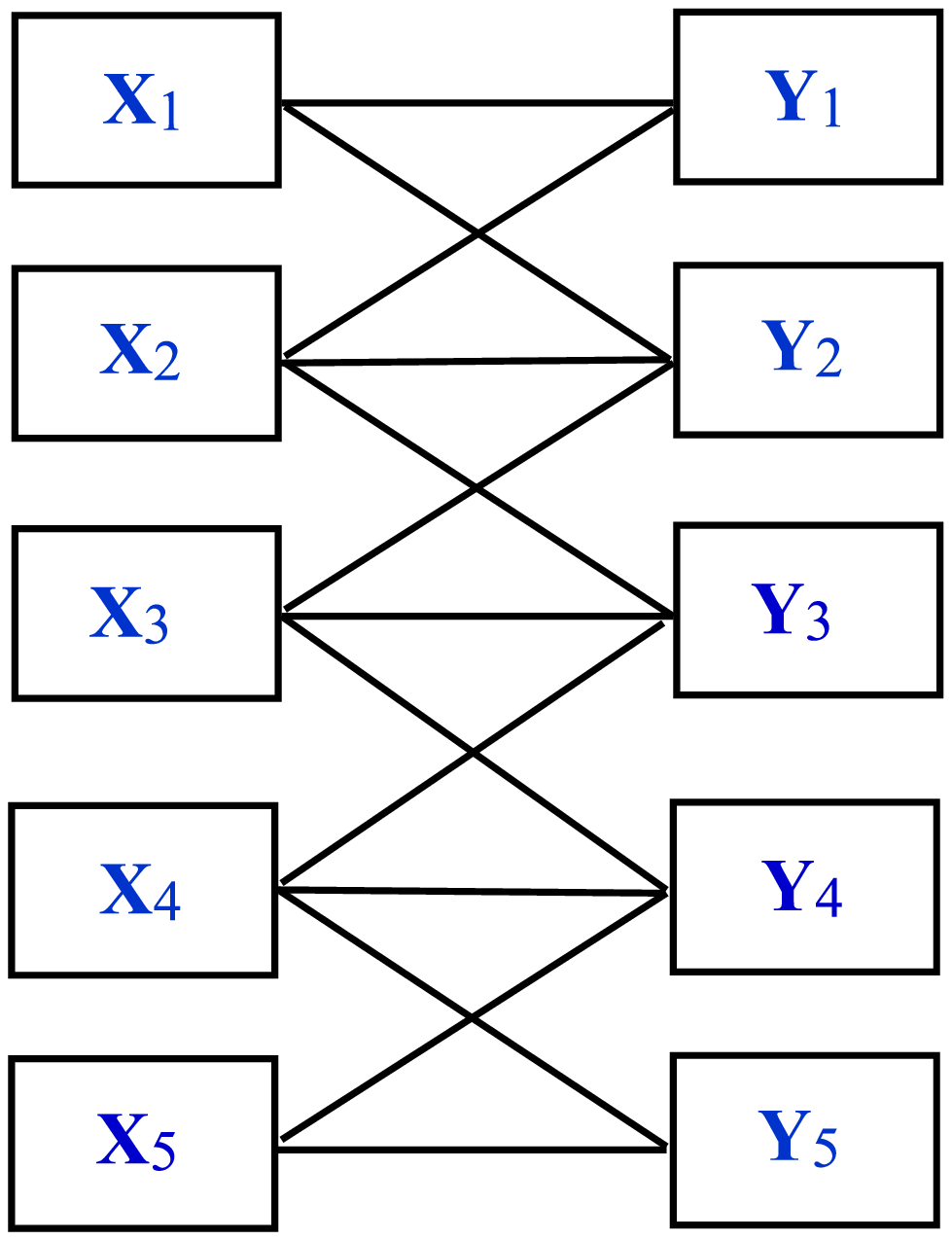}}
  \caption{Figure showing examples of the considered channel models with a number of users $K=5$. In $(a)$, a fully connected channel model is shown. In ($b$), a locally connected channel model with connectivity parameter $L=2$ is shown.}
  \label{fig:channelmodel}
\end{figure}

\subsection{Cooperation Model}

For each $i \in [K]$, let ${\cal T}_i \subseteq [K]$ be the transmit set of receiver $i$, i.e., those transmitters with the knowledge of $W_i$. The transmitters in ${\cal T}_i$ cooperatively transmit the message $W_i$ to the receiver $i$. The messages $\{W_i\}$ are assumed to be independent of each other. The \emph{cooperation order} $M$ is defined to be the maximum transmit set size,
\begin{equation}\label{eq:coop_order}
M = \max_i |{\cal T}_i|.
\end{equation}
For any set ${\cal A} \subseteq [K]$, we define $C_{\cal A}$ as the set of messages carried by transmitters with indices in ${\cal A}$, i.e., the set $\{i: {\cal T}_i \cap {\cal A} \neq \phi\}$.

\subsection{Degrees of Freedom}\label{sec:dofmodel}
Let $P$ be the average transmit power constraint at each transmitter, and let ${\cal W}_i$ denote the alphabet for message $W_i$. Then the rates $R_i(P) = \frac{\log|{\cal W}_i|}{n}$ are achievable if the decoding error probabilities of all messages can be simultaneously made arbitrarily small for large enough $n$, and this holds for almost all channel realizations. The degrees of freedom $d_i, i\in[K],$ are defined as $d_i=\lim_{P \rightarrow \infty} \frac{R_i(P)}{\log P}$. The DoF region ${\cal D}$ is the closure of the set of all achievable DoF tuples. The total number of degrees of freedom ($\eta$) is the maximum value of the sum of the achievable degrees of freedom, $\eta=\max_{\cal D} \sum_{i \in [K]} d_i$.

For a $K$-user channel, we define $\eta(K,M)$ as the best achievable $\eta$ over all choices of transmit sets satisfying the cooperation order constraint in \eqref{eq:coop_order}. Similarly, we define $\eta_L(K,M)$ for a locally connected channel with $L$ interfering signals per receiver. 

In order to simplify our analysis, we define the asymptotic per user DoF $\tau(M)$, and $\tau_L(M)$ to measure how $\eta(K,M)$, and $\eta_L(K,M)$ scale with $K$, respectively, while all other parameters are fixed,
\begin{equation}
\tau(M) = \lim_{K\rightarrow \infty} \frac{\eta(K,M)}{K},
\end{equation}
\begin{equation}
\tau_L(M) = \lim_{K\rightarrow \infty} \frac{\eta_L(K,M)}{K}.
\end{equation}

For the locally connected channel model where $L>1$, let $x=\left \lfloor \frac{L}{2} \right \rfloor$. We silence the first $x$ transmitters, deactivate the last $x$ receivers, and relabel the transmit signals to obtain a ($K-x$)-user channel, where the transmitter $i$ is connected to receivers in the set $\{Y_k: k\in\{i,i+1,\ldots,i+L\}\}$. We note that the new channel model gives the same value of $\tau_L(M)$ as the original one, since $x = o(K)$. Unless explicitly stated otherwise, we will be using this equivalent model in the rest of the paper. We show an example construction of the equivalent channel model in Figure~\ref{fig:locallyconnected}.

\begin{figure}
  \centering
\subfloat[]{\label{fig:fc}\includegraphics[height=0.25\textwidth]{locallyconnected.eps}}                
\quad\quad\quad\quad\subfloat[]{\label{fig:lc}\includegraphics[width=0.188\textwidth]{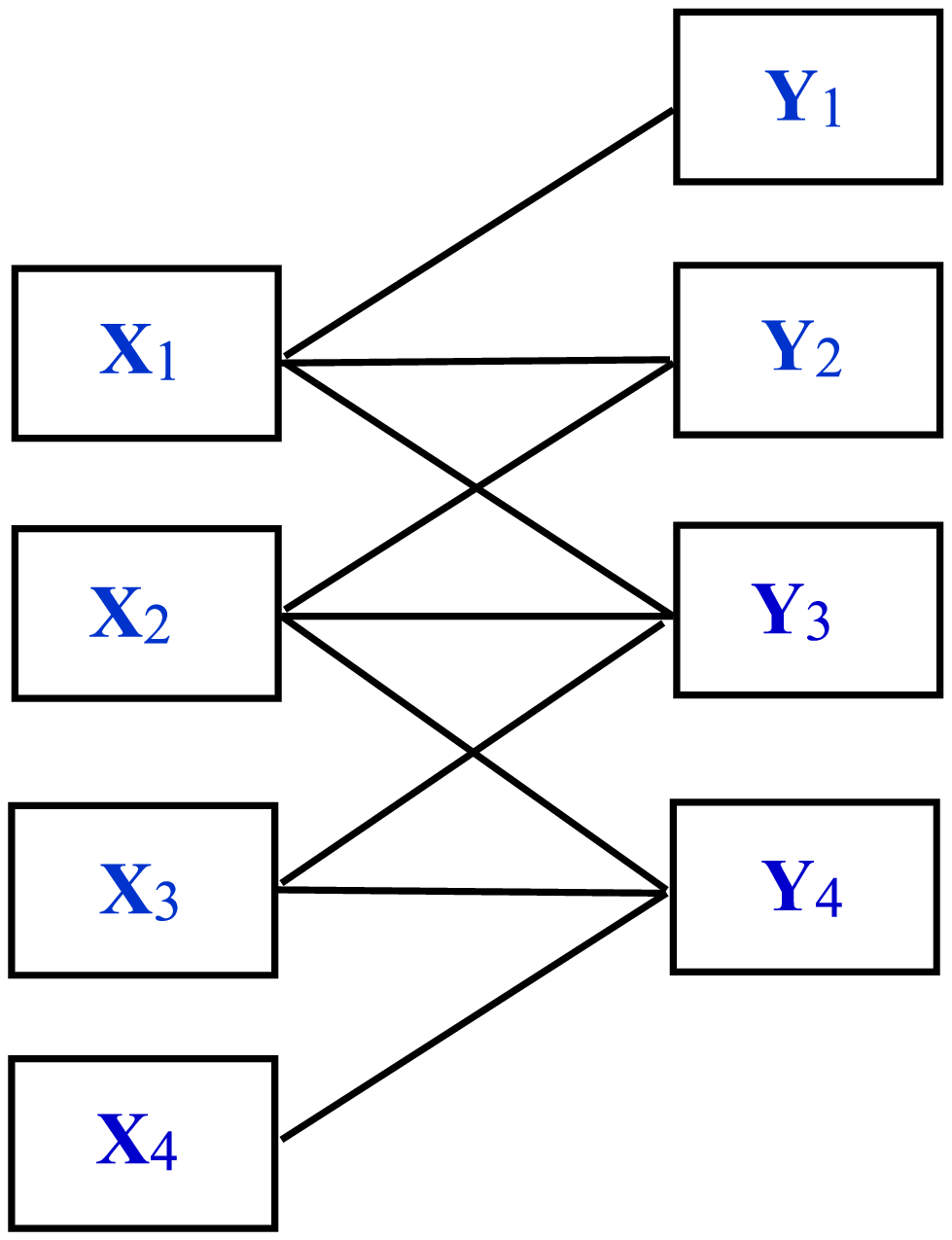}}
  \caption{Construction of the equivalent locally connected channel model with a number of users $K=5$ and connectivity parameter $L=2$. In $(a)$, the original model of~\eqref{eq:channel} is shown. In ($b$), the new model is shown.}
  \label{fig:locallyconnected}
\end{figure}

\subsection{Message Assignment Strategy}

A message assignment strategy is defined by a sequence of transmit sets $({\cal T}_{i,K}), i\in[K], K\in\{1,2,\ldots\}$. For each positive integer $K$ and $\forall i\in[K]$,  ${\cal T}_{i,K} \subseteq [K], |{\cal T}_{i,K}| \leq M$. We use message assignment strategies to define the transmit sets for a sequence of $K-$user channels. The $k^{\mathrm{th}}$ channel in the sequence has $k$ users, and the transmit sets for this channel are defined as follows. The transmit set of receiver $i$ in the $k^{\mathrm{th}}$ channel in the sequence is the transmit set ${\cal T}_{i,k}$ of the message assignment strategy. 

We call a message assignment strategy \emph{optimal} for a sequence of $K-$user fully connected channels, $K\in\{1,2,\ldots\}$,  if and only if there exists a sequence of coding schemes achieving $\tau(M)$ using the transmit sets defined by the message assignment strategy. A similar definition applies for locally connected channels.

\subsection{Local Cooperation}\label{sec:localcooperation}

We say that a message assignment strategy satisfies the local cooperation constraint, if and only if there exists a function $r(K)$ such that $r(K)=o(K)$, and
\begin{equation}\label{eq:localcooperation}
{\cal T}_{i,K} \subseteq \{i-r(K),i-r(K)+1,\ldots,i+r(K)\}, \forall i \in [K], \forall K\in {\bm Z}^{+}.
\end{equation}
Let $\tau^{\scriptscriptstyle \mathrm{(loc)}}(M)$ and $\tau_L^{\scriptscriptstyle \mathrm{(loc)}}(M)$ be the maximum achievable asymptotic per user DoF $\tau(M)$ and $\tau_L(M)$ under the additional local cooperation constraint, respectively.

\section{Informal Summary of Results}\label{sec:informalsummary}
In this paper, we study the benefit of CoMP transmission via the asymptotic per user DoF. In particular, we investigate whether the asymptotic per user DoF increases by allowing CoMP transmission, i.e., by allowing a cooperation order $M > 1$, and characterize this improvement, if it exists, as a function of $M$. 

The considered problem is completely described by two system parameters, namely, the channel connectivity and the cooperation order $M$. We attempt to find an answer by setting two design parameters: the message assignment strategy satisfying the cooperation order constraint, and the achievable scheme.  

%IT 11 dof gain
%no asymptotic per user dof gain
%local cooperation
%M=2 negative result

We know from the results in~\cite{Madsen-Nosratinia} and~\cite{Cadambe-IA} that the asymptotic per user DoF of the fully connected channel is $\frac{1}{2}$ if each message is available only at the transmitter carrying the same index; it is straightforward to extend this result to the case where each message can be available at any single transmitter, i.e., the case where $M=1$, and hence, we know that $\tau(M=1)=\frac{1}{2}$. In~\cite{Annapureddy-ElGamal-Veeravalli-IT11}, it was shown that CoMP transmission achieves a DoF gain for the fully connected channel. However, this gain does not scale linearly with the number of users $K$. The considered message assignment strategy in~\cite{Annapureddy-ElGamal-Veeravalli-IT11} is the spiral strategy where each message is assigned to the transmitter carrying the same index as well as $M-1$ succeeding transmitters. We note that this strategy satisfies the local cooperation constraint defined in Section~\ref{sec:localcooperation}, and show in Section~\ref{sec:asymptoticpudof} that local cooperation cannot achieve an asymptotic per user DoF gain for the fully connected channel. More precisely, we show that
\begin{equation}
\tau^{(\textrm{loc})}(M)=\tau(1)=\frac{1}{2}, \forall M.
\end{equation}
Furthermore, we extend this negative conclusion in Section~\ref{sec:asymptoticpudof} to all message assignments that are restricted to assign each message to at most two transmitters, i.e.,
\begin{equation}
\tau(M=2)=\tau(M=1)=\frac{1}{2}.
\end{equation}
In general, we prove in Theorem~\ref{thm:tauouterbound} the following upper bound on the asymptotic per user DoF for any value of $M$,
\begin{equation}\label{eq:summarytaubound}
\tau(M) \leq \frac{M-1}{M}.
\end{equation}
We then show that the tightness of~\eqref{eq:summarytaubound} for the case where $M=2$ does not generalize. In particular, we prove in Theorem~\ref{thm:taufthree} the following tighter bound for the case where $M=3$,
\begin{equation}
\tau(M=3) \leq \frac{5}{8}.
\end{equation}
We summarize the results obtained for the fully connected channel model in Table~\ref{tab:fc}.
\begin{table}
\begin{center}
    \begin{tabular}{ | l | l | l | l | l |}
    \hline
     & $M=1$ & $M=2$ & $M=3$ & General Values of $M$ \\ \hline
    Local Cooperation & $\tau^{(\textrm{loc})}(M=1)=\frac{1}{2}$ & $\tau^{(\textrm{loc})}(M=2)=\frac{1}{2}$ & $\tau^{(\textrm{loc})}(M=3)=\frac{1}{2}$ & $\tau^{(\textrm{loc})}(M)=\frac{1}{2}$ \\ \hline
    General Cooperation & $\tau(M=1)=\frac{1}{2}$ & $\tau(M=2)=\frac{1}{2}$ & $\tau(M=3)\leq\frac{5}{8}$ & $\tau(M) \leq \frac{M-1}{M}$ \\ \hline
    \end{tabular}
\end{center}
\caption{Summary of results for fully connected channels.}
\label{tab:fc}
\end{table}

%Locally connected IC results
In Section~\ref{sec:lc}, we illustrate how the result introduced in~\cite{Lapidoth-Shamai-Wigger-ISIT07} shows that asymptotic per user DoF gains are possible for the special case of the locally connected channel where each transmitter is connected to the receiver with the same index as well as one succeeding receiver ($L=1$). In particular, the enabling message assignment strategy is the spiral strategy that satisfies the local cooperation constraint. We then introduce in Section~\ref{sec:dofgains} a simple zero-forcing transmit beamforming scheme that achieves a higher asymptotic per user DoF than that shown in~\cite{Lapidoth-Shamai-Wigger-ISIT07}. In particular, we show that
\begin{equation}
\tau_L(M) \geq \max\left\{\frac{1}{2},\frac{2M}{2M+L}\right\}, \forall M,L.
\end{equation}
Moreover, this lower bound is optimal if we restrict ourselves to the class of schemes that satisfy an interference avoidance constraint. We then provide an upper bound in Section~\ref{sec:lcupperbound} that completes the characterization of the asymptotic per user DoF for the case where $L=1$, i.e., showing that
\begin{equation}\label{eq:loneresult}
\tau_1(M)=\frac{2M}{2M+1}, \forall M.
\end{equation}
In particular, the optimal message assignment strategy for the case where $L=1$ satisfies a local cooperation constraint. We show in Section~\ref{sec:lcusefulmsgassignment} that local cooperation is optimal for all locally connected channels, thereby establishing that the negative result regarding local cooperation for the fully connected channel is due only to the assumption of full connectivity.

We note that~\eqref{eq:loneresult} implies that the asymptotic per user DoF for Wyner's asymmetric model is strictly greater than $\frac{1}{2}$ even for the case of \emph{no cooperation} (i.e., $M=1$). We show however in Section~\ref{sec:lcupperbound} that this is only the case for $L=1$, and does not hold for all other locally connected channels. We summarize the results obtained for the locally connected channel model in Table~\ref{tab:lc}.
\begin{table}
\begin{center}
    \begin{tabular}{ | l | l | l | l | l |}
    \hline
     & $M=1$ & $M=2$ & General Values of $M$ \\ \hline
    $L=1$ & $\tau_1(M=1)=\frac{2}{3}$ & $\tau_1(M=2)=\frac{4}{5}$ & $\tau_1(M)=\frac{2M}{2M+1}$ \\ \hline
    $L=2$ & $\tau_2(M=1)=\frac{1}{2}$ & $\tau_2(M=2)\geq\frac{2}{3}$ & $\tau_2(M) \geq \frac{M}{M+1}$ \\ \hline
    General Values of $L$ & $\tau_L(1)=
\begin{cases}
\frac{2}{3},\quad  &\text{if} \quad L=1,\\
\frac{1}{2},\quad &\text{if} \quad L \geq 2.
\end{cases}$ & $\tau_L(M=2)\geq\max\left\{\frac{1}{2},\frac{4}{4+L}\right\}$ & $\tau_L(M) \geq \max\left\{\frac{1}{2},\frac{2M}{2M+L}\right\}$ \\ \hline
    \end{tabular}
\end{center}
\caption{Summary of results for locally connected channels.}
\label{tab:lc}
\end{table}

%Proof techniques
%vertex expander message assignment strategies
%useful message assignments for locally connected channels

\subsection{Proof Techniques}
In Section~\ref{sec:fcupperbound}, we restate from~\cite{Annapureddy-ElGamal-Veeravalli-IT11} a necessary condition on any point in the DoF region of the fully connected channel with a fixed message assignment. We then use this condition to derive Corollary~\ref{cor:dofouterbound} that implies directly all of the provided DoF upper bounds for the fully connected channel. Moreover, Corollary~\ref{cor:dofouterbound} can be used to shed insight on the open problem of determining whether $\tau(M)>\frac{1}{2}$ for $M\geq3$. In~\cite{ElGamal-Annapureddy-Veeravalli-CISS12}, we showed that if Corollary~\ref{cor:dofouterbound} is tight for all instances of the problem, then scalable DoF cooperation gains are achievable for $M \geq 3$, i.e., it would follow that $\tau(M)>\frac{1}{2}$ for $M\geq3$. 

It is obvious for locally connected channels that some message assignments are \emph{reducible}. For example, for the case where $M=1$, any assignment of a message $W_i$ to a transmitter that is not connected to the $i^{\mathrm{th}}$ receiver cannot achieve a positive rate for communication of that message. To prove DoF upper bounds for locally connected channels with CoMP transmission, we use a characterization of necessary conditions on \emph{irreducible message assignments}, as discussed in Section~\ref{sec:lcusefulmsgassignment}.

\section{Fully Connected Interference Channel}\label{sec:fc}

 In this section, we investigate whether $\tau(M) > \frac{1}{2}$ for $M > 1$, and message assignment strategies that may lead to a positive conclusion. 

\subsection{Prior Work}\label{sec:fcpriorwork}
We know from~\cite{Madsen-Nosratinia}, and~\cite{Cadambe-IA} that the per user DoF of a fully connected interference channel without cooperation is $\frac{1}{2}$, i.e., $\tau(1)=\frac{1}{2}$. In~\cite{Annapureddy-ElGamal-Veeravalli-IT11}, the following spiral message assignment strategy was considered for $M \geq 1$:
\vspace{5 mm}
${\cal T}_{i,K}=
\begin{cases}
\{i,i+1,\ldots,i+M-1\}, \quad &\forall i \in [K-(M-1)],\\
\{i,i+1,\ldots,K,1,2,\ldots,M-(K-i+1)\},\quad &\forall  i \in \{K-(M-2),K-(M-2)+1,\ldots,K\},
\end{cases}$

Using this message assignment strategy and an asymptotic interference alignment scheme, it was shown in \cite[Theorem $5$]{Annapureddy-ElGamal-Veeravalli-IT11} that

\begin{equation}
\eta(K,M) \geq \frac{K+M-1}{2}, \forall M \leq K < 10,
\end{equation}
and it was shown in \cite{Annapureddy-ElGamal-Veeravalli-IT11} that this lower bound is within one degree of freedom from the maximum achievable DoF using the spiral message assignment strategy. However, even if $\eta(K,M)=\frac{K+M-1}{2}$ for all values of $K$, the DoF gain due to CoMP transmission (beyond $\frac{K}{2}$) does not scale with the number of users $K$. Hence, the question of whether $\tau(M) > \frac{1}{2}$ for $M>1$ remains open. Here, we note that the spiral message assignment strategy satisfies the local cooperation constraint and in Section~\ref{sec:asymptoticpudof}, we generalize the negative conclusion of~\cite{Annapureddy-ElGamal-Veeravalli-IT11} to all message assignment strategies satisfying the local cooperation constraint.

\subsection{DoF Upper Bound}\label{sec:fcupperbound}

In order to characterize the DoF of the channel $\tau(M)$, we need to consider all possible strategies for message assignments satisfying the cooperation order constraint defined in~\eqref{eq:coop_order}. Through the following corollary of~\cite[Theorem $1$]{Annapureddy-ElGamal-Veeravalli-IT11}, we provide a way to upper bound the maximum achievable DoF for each such assignment, thereby, introducing a criterion for comparing different message assignments satisfying~\eqref{eq:coop_order} using the special cases where this bound holds tightly. We first restate the following result from~\cite{Annapureddy-ElGamal-Veeravalli-IT11}.

\begin{lem} [\cite{Annapureddy-ElGamal-Veeravalli-IT11}] \label{thm:dofouterbound}
Any point $(d_1,d_2,\ldots,d_K)$ in the DoF region of a $K-$user fully connected channel satisfies the inequalities:
\begin{equation}\label{eq:comp-dof-ob-region} 
\begin{split} 
 \sum_{k: {\cal T}_{k} \subseteq {\cal A} \text{ or } k \in {\cal B}} d_{k} \leq \max(|{\cal A}|,|{\cal B}|),  \forall {\cal A}, {\cal B} \subseteq [K].
\end{split} 
\end{equation}
\end{lem}
\begin{proof}
The proof is availble in~\cite[Theorem $1$]{Annapureddy-ElGamal-Veeravalli-IT11}.
\end{proof}
We now provide the following corollary for bounding the DoF number $\eta$ of a $K-$user fully connected channel with a fixed message assignment. Recall that for a set of transmitter indices ${\cal S}$, the set $C_{\cal S}$ is the set of messages carried by transmitters in ${\cal S}$. 

\begin{cor}\label{cor:dofouterbound}
For any $m,\bar{m} : m+\bar{m} \geq K$, if there exists a set ${\cal S}$ of indices for transmitters carrying no more than $m$ messages, and $|{\cal S}|=\bar{m}$, then $\eta \leq m$, or more precisely, 
\begin{equation}
\eta \leq \min_{{\cal S} \subseteq [K]}\max (|C_{\cal S}|, K-|{\cal S}|).
\end{equation} 
\end{cor}
\begin{proof}
For each subset of transmitter indices ${\cal S} \subseteq [K]$, we apply Lemma~\ref{thm:dofouterbound} with the sets ${\cal A}$ and ${\cal B}$ assigned as follows, ${\cal A}=\bar{\cal S}$ and ${\cal B}=C_{\cal S}$. 
\end{proof}

We refer the reader to Figure~\ref{fig:dofouterbound} for an example illustration of Corollary~\ref{cor:dofouterbound}.
\begin{figure}
\centering
\includegraphics[width=0.3\columnwidth]{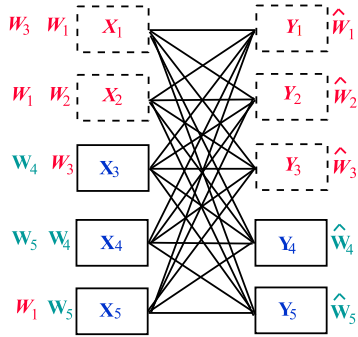}
\caption{Example application of Corollary~\ref{cor:dofouterbound}, with ${\cal S}=\{1,2\}$ and $C_{\cal S}=\{1,2,3\}$. Transmit signals with indices in ${\cal S}$, and messages as well as receive signals with indices in $C_{\cal S}$ are shown in tilted red font and dashed boxes. The DoF $\eta \leq |C_{\cal S}|=K-|{\cal S}|=3$.}
\label{fig:dofouterbound}
\end{figure}

\subsection{Asymptotic DoF Cooperation Gain}\label{sec:asymptoticpudof}
We now use Corollary~\ref{cor:dofouterbound} to prove upper bounds on the asymptotic per user DoF $\tau(M)$. 

In an attempt to reduce the complexity of the problem of finding an optimal message assignment strategy, we begin by considering message assignment strategies satisfying the local cooperation constraint defined in Section~\ref{sec:localcooperation}. We now show that a scalable cooperation DoF gain cannot be achieved using local cooperation.
  
\begin{thm}\label{thm:localcooperation}
Any message assignment strategy satisfying the local cooperation constraint of~\eqref{eq:localcooperation} cannot be used to achieve an asymptotic per user DoF greater than that achieved without cooperation. More precisely,
\begin{equation}
\tau^{(\scriptscriptstyle \mathrm{loc})}(M) = \frac{1}{2}, \text{ for all } M.
\end{equation}
\end{thm}
\begin{proof}
Fix $M \in {\bm Z}^{+}$. For any value of $K\in {\bm Z}^{+}$, we use Corollary~\ref{cor:dofouterbound} with the set ${\cal S}=\{1,2,\ldots,\left \lceil \frac{K}{2} \right \rceil \}$. Note that $C_{\cal S} \subseteq \{1,2,\ldots, \left \lceil \frac{K}{2} \right \rceil + r(K)\}$, and hence, it follows that $\eta(K,M) \leq \left \lceil \frac{K}{2} \right \rceil + r(K)$. Finally, $\tau(M) = \lim_{K \rightarrow \infty} \frac{\eta(K,M)}{K} \leq \frac{1}{2}$. The lower bound follows from~\cite{Cadambe-IA} without cooperation.
\end{proof}

We now investigate if it is possible for the cooperation gain to scale linearly with $K$ for fixed $M$.  It was shown in Theorem~\ref{thm:localcooperation} that such a gain is not possible for message assignment strategies that satisfy the local cooperation constraint. Here, we only impose the cooperation order constraint in~\eqref{eq:coop_order} and prove in Theorem~\ref{thm:tauouterbound} an upper bound on $\tau(M)$ that is tight enough for finding $\tau(2)$. 

\begin{thm}\label{thm:tauouterbound}
For any cooperation order constraint $M \geq 2$, the following upper bound holds for the asymptotic per user DoF,
\begin{equation}
\tau(M) \leq \frac{M-1}{M}.
\end{equation}
\end{thm}
\begin{proof}
For any value of $M$ and $K$, we show that $\eta(K,M) \leq \frac{K(M-1)}{M} + o(K)$. For every value of $K$ such that $\frac{K-1}{M}$ is an integer, we show that $\eta(K,M) \leq \frac{K(M-1)+1}{M}$. When $\frac{K-1}{M}$ is not an integer, we add $x=o(K)$ extra users such that $\frac{K+x-1}{M}$ is an integer, and bound the DoF as follows,
\begin{eqnarray}
\eta(K,M) &\leq& \eta(K+x,M)
\\&\leq& \frac{(K+x)(M-1)+1}{M}
\\&=&\frac{K(M-1)}{M} + o(K)
\end{eqnarray}
It then suffice to consider the case where $\frac{K-1}{M}$ is an integer. The idea is to show that for any assignment of messages satisfying the cooperation order constraint, there exists a set of indices ${\cal S} \subseteq [K]$ for $\frac{K-1}{M}$ transmitters that do not carry more than $K-\frac{K-1}{M}$ messages, and then the DoF upper bound follows by applying Corollary~\ref{cor:dofouterbound}. More precisely, it suffices to show that the following holds,
\begin{equation}\label{eq:dofouterboundproof}
\forall K: \frac{K-1}{M}\in{\bm Z}^+, \exists {\cal S} \subseteq [K]: |{\cal S}|=\frac{K-1}{M}, |{\cal C}_{\cal S}|=\frac{K(M-1)+1}{M}=K-|{\cal S}|.
\end{equation}
We first illustrate simple examples that demonstrate the validity of~\eqref{eq:dofouterboundproof}. Consider the case where $K=3$, $M=2$, we need to show in this case that there exists a transmitter that does not carry more than two messages, which follows by the pigeonhole principle since each message can only be available at a maximum of two transmitters. Now, consider the slightly more complex example of $K=5$, $M=2$, we need to show in this case that there exists a set of two transmitters that do not carry more than three messages. We know that there is a transmitter carrying at most two messages, and we select this transmitter as the first element of the desired set. Without loss of generality, let the two messages available at the selected transmitter be $W_1$ and $W_2$. Now, we need to find another transmitter that carries at most one message among the messages in the set $\{W_3,W_4,W_5\}$. Since each of these three messages can be available at a maximum of two transmitters, and we have four transmitters to choose from, one of these transmitters has to carry at most one of these messages. By adding the transmitter satisfying this condition as the second element of the set, we obtain a set of two transmitters carrying no more than three messages, and~\eqref{eq:dofouterboundproof} holds. 

We extend the argument used in the above examples through Lemmas~\ref{lem:basis} and~\ref{lem:inductionstep} in the Appendix. We know by induction using these lemmas that~\eqref{eq:dofouterboundproof} holds, and the theorem statement follows. 

%The asymptotic bound then follows by adding $o(K)$ extra users whenever needed such that the integer constraint is always satisfied. 
%We show that the following stronger statement holds,
%
%\begin{equation}
%\eta(K,M) \leq \frac{K(M-1)+M+1}{M}, \forall M \geq 2.
%\end{equation}
%
%Assume that $n=\frac{K-1}{M}$ is an integer. We know by induction from Lemmas~\ref{lem:basis} and~\ref{lem:inductionstep} in the Appendix that $\exists S \subset [K]$, $|S|=n$, $|C_{\cal S}| \leq (M-1)n+1=\frac{K(M-1)+1}{M}=K-|{\cal S}|$. Now, applying Corollary~\ref{cor:dofouterbound} proves that $\eta(K,M) \leq \frac{K(M-1)+1}{M}$. For the case where $\frac{K-1}{M}$ is not an integer, let $x$ be the largest integer less than $K$ such that $\frac{x-1}{M}$ is an integer. Now, we ignore the last $K-x$ users and bound the sum DoF for the remaining users by $\frac{x(M-1)+1}{M}$ to show that $\eta(K,M) \leq \frac{x(M-1)+1}{M}+(K-x)$, and hence,
%\begin{eqnarray}
%\eta(K,M) &\leq& \frac{x(M-1)+1}{M}+(K-x)\nonumber
%\\&=& \frac{K(M-1)+1}{M}+\frac{K-x}{M}\nonumber
%\\&\leq& \frac{K(M-1)+1}{M}+1.
%\end{eqnarray}  
\end{proof}

Together with the achievability result in~\cite{Cadambe-IA}, the statement in Theorem~\ref{thm:tauouterbound} implies the following corollary.

\begin{cor}\label{cor:tauftwo}
For any message assignment strategy such that each message is available at a maximum of two transmitters, the asymptotic per user DoF is the same as that achieved without cooperation. More precisely,
\begin{equation}
\tau(2) = \frac{1}{2}.
\end{equation}
\end{cor}

%state result \tau(3) < 5/8 (we don't know whether \tau(M)=1/2 for all M or not)
The characterization of $\tau(M)$ for values of $M > 2$ remains an open question, as Theorem~\ref{thm:tauouterbound} is only an upper bound. Moreover, the following result shows that the upper bound in Theorem~\ref{thm:tauouterbound} is loose for $M=3$. 

\begin{thm}\label{thm:taufthree}
For any message assignment strategy such that each message is available at a maximum of three transmitters, the following bound holds for the asymptotic per user DoF,
\begin{equation}\label{eq:taufthree}
\tau(3) \leq \frac{5}{8}.
\end{equation}
\end{thm}
\begin{proof}
In a similar fashion to the proof of Theorem~\ref{thm:tauouterbound}, we prove the statement by induction. The idea is to prove the existence of a set ${\cal S}$ with approximately $\frac{3K}{8}$ transmitter indices, and these transmitters are carrying no more than approximately $\frac{5K}{8}=K-|{\cal S}|+o(K)$ messages, and then use Corollary~\ref{cor:dofouterbound} to derive the DoF outer bound. In the proof of Theorem~\ref{thm:tauouterbound}, we used Lemmas~\ref{lem:basis} and~\ref{lem:inductionstep} in the Appendix, to provide the basis and induction step of the proof, respectively. Here, we follow the same path until we show that there exists a set ${\cal S}$ such that $|{\cal S}|=\frac{K+1}{4}$ and $|C_{\cal S}|\leq (M-1)|{\cal S}|+1$, and then we use Lemma~\ref{lem:mthreeinductionstep} in the Appendix to provide a stronger induction step that establishes a tighter bound on the size of the set $C_{\cal S}$.  

We note that it suffices to show that $\eta(K,3) \leq \frac{5K}{8}+o(K)$ for all values of $K$ such that $\frac{K+1}{4}$ is an even positive integer, and hence, we make that assumption for $K$. Define the following,
\begin{equation}
x_1 = \frac{K+1}{4},
\end{equation}
\begin{equation}
x_2=\frac{K-7}{8},
\end{equation}
\begin{equation}
x_3 = 2 x_1 + 1 + x_2.
\end{equation}
Now, we note that,
\begin{equation}\label{eq:zequality}
x_3 = K -  (x_1 + x_2) ,
\end{equation}
and by induction, it follows from Lemmas~\ref{lem:basis} and~\ref{lem:inductionstep} that $\exists {\cal S}_1 \subset [K]$, $|{\cal S}_1|= x_1$, $|C_{{\cal S}_1}| \leq 2x_1+1$. We now apply induction again with the set ${\cal S}_1$ as a basis, and use Lemma~\ref{lem:mthreeinductionstep} for the induction step to show that $\exists {\cal S}_2 \subset [K]$, $|{\cal S}_2|= x_1+x_2$, $|C_{{\cal S}_2}| \leq x_3=K-|{\cal S}_2|$. Hence, we get the following upper bound using Corollary~\ref{cor:dofouterbound},
\begin{eqnarray}
\eta(K,3) &\leq& x_3\nonumber 
\\&=&\frac{5(K+1)}{8},
\end{eqnarray}
from which~\eqref{eq:taufthree} holds.
\end{proof}

We note that all the DoF upper bounding proofs used so far employ Corollary~\ref{cor:dofouterbound}. In~\cite{ElGamal-Annapureddy-Veeravalli-CISS12}, we showed that under the hypothesis that the upper bound in Corollary~\ref{cor:dofouterbound} is tight for any $K$-user fully connected interference channel with a cooperation order constraint $M$, then scalable DoF cooperation gains are achievable for any value of $M \geq 3$. Hence, a solution to the general problem necessitates the discovery of either new upper bounding techniques or new coding schemes. However, as we show next, scalable DoF cooperation gains are possible when the assumption of full connectivity is relaxed.

\section{Locally Connected Interference Channels}\label{sec:lc}

In Section~\ref{sec:channelmodel}, we defined the locally connected channel model as a function of the number of dominant interferers per receiver $L$, by connecting each transmitter to $\left \lfloor \frac{L}{2} \right \rfloor$ preceding receivers and $\left \lceil \frac{L}{2} \right \rceil$ succeeding receivers, and in Section~\ref{sec:dofmodel}, we illustrated an equivalent model in terms of the asymptotic per user DoF $\tau_L(M)$. In the equivalent model, each transmitter is connected to $L$ succeeding receivers. More precisely, we consider the following channel model,
\begin{equation}\label{eq:equivalentchannel}
H_{ij} \text{ is not identically } 0 \text { if and only if } j \in \left[i, i+1,\ldots,i+L\right],
\end{equation}
and all non-zero channel coefficients are generic.

\subsection{Prior Work}\label{sec:lcpriorwork}
In~\cite{Lapidoth-Shamai-Wigger-ISIT07}, the special case of Wyner's asymmetric model ($L=1$) was considered, and the spiral message assignment strategy mentioned in Section~\ref{sec:fcpriorwork} was fixed, i.e., each message is assigned to its own transmitter as well as $M-1$ following transmitters. The asymptotic per user DoF was then characterized as $\frac{M}{M+1}$. This shows for our problem that,
\begin{equation}\label{eq:tauwiggerone}
\tau_1(M) \geq \frac{M}{M+1}.
\end{equation}
In \cite[Remark $2$]{Shamai-Wigger-ISIT11}, a message assignment strategy was described to enable the achievability of an asymptotic per user DoF as high as $\frac{2M-1}{2M}$, it can be easily verified that this is indeed true, and hence, we know that,
\begin{equation}\label{eq:tauwiggertwo}
\tau_1(M) \geq \frac{2M-1}{2M}.
\end{equation} 
The main difference in the strategy described in \cite[Remark $2$]{Shamai-Wigger-ISIT11} from the spiral message assignment strategy considered in \cite{Lapidoth-Shamai-Wigger-ISIT07}, is that unlike the spiral strategy, messages are assigned to transmitters in an asymmetric fashion, where we say that a message assignment is symmetric if and only if for all $j,i\in[K], j > i$, the transmit set ${\cal T}_j$ is obtained by shifting forward the indices of the elements of the transmit set ${\cal T}_i$ by $(j-i)$.

We show that both the message assignment startegy analyzed in~\cite{Lapidoth-Shamai-Wigger-ISIT07} and the one suggested in~\cite{Shamai-Wigger-ISIT11} are suboptimal for $L=1$, and the value of $\tau_1(M)$ is in fact strictly higher than the bounds in~\eqref{eq:tauwiggerone} and~\eqref{eq:tauwiggertwo}. The key idea enabling our result is that each message need not be available at the transmitter carrying its own index. We start by illustrating a simple example for the case of no cooperation ($M=1$) that highlights the idea behind our scheme. 

\subsection{Example: $M=L=1$}\label{sec:lcexample}
\begin{figure}
\centering
\includegraphics[width=0.5\columnwidth]{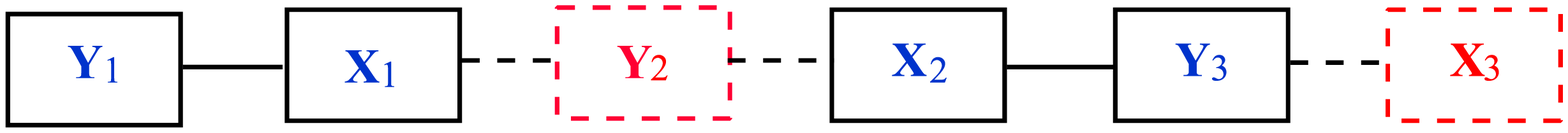}
\caption{Achieving $2/3$ per user DoF for $M=L=1$. Each transmitter is carrying a message for the receiver connected to it by a solid line. Figure showing only signals corresponding to the first $3$ users in a general $K-$user network. Signals in dashed boxes are deactivated. Note that the deactivation of $X_3$ splits this part of the network from the rest.}
\label{fig:examplemone}
\end{figure}

Let $W_1$ be available at the first transmitter, $W_3$ be available at the second transmitter, and deactivate both the second receiver and the third transmitter. Then it is easily seen that messages $W_1$ and $W_3$ can be received without interfering signals at their corresponding receivers. Moreover, the deactivation of $X_3$ splits this part of the network from the rest. i.e., the same scheme can be repeated by assigning $W_4, W_6$, to the transmitters with transmit signals $X_4, X_5$, respectively, and so on. Thus, $2$ degrees of freedom can be achieved for each set of $3$ users, thereby, achieving an asymptotic per user DoF of $\frac{2}{3}$. The described message assignment is depicted in Figure~\ref{fig:examplemone}. It is evident now that a constraint that is only a function of the load on the {\em backhaul} link may lead to a discovery of better message assignments than the one considered in~\cite{Lapidoth-Shamai-Wigger-ISIT07}. In the following section, we show that the optimal message assignment strategy under the  cooperation order constraint~\eqref{eq:coop_order} is different from the spiral strategy. The above described message assignment strategy and coding scheme for the special case of $M=L=1$ are shown to be optimal. It is worth noting that the optimality of a TDMA scheme in this case follows as a special case from a general result in~\cite{Maleki-Jafar-arXiv13}, where necessary and sufficient conditions on channel connectivity and message assignment are derived for TDMA schemes to be optimal.

\subsection{Achieving Scalable DoF Cooperation Gains}\label{sec:dofgains}
In this Section, we specialize the scheme introduced for multiple-antenna transmitters in~\cite[Section IV]{ElGamal-Annapureddy-Veeravalli-ISIT12} to our setting. We consider a simple linear precoding coding scheme, where each message is assigned to a set of transmitters with successive indices, and a zero-forcing transmit beamforming strategy is employed. The transmit signal at the $j^{\mathrm{th}}$ transmitter is given by,
\begin{equation}\label{eq:linear-precoding}
X_j = \sum_{i: j \in {\cal T}_i} X_{j,i},
\end{equation} 
where $X_{j,i}$ depends only on message $W_i$.

Using simple zero-forcing transmit beams with a fractional reuse scheme that activates only a subset of transmitters and receivers in each channel use, we extend the example in Section~\ref{sec:lcexample} to achieve scalable DoF cooperation gains for any value of $M > \frac{L}{2}$.

\begin{thm}\label{thm:lowerbound}
The following lower bound holds for the asymptotic per user DoF of a locally connected channel with connectivity parameter $L$,
\begin{equation}\label{eq:lowerbound}
\tau_L(M) \geq \max\left\{\frac{1}{2},\frac{2M}{2M+L}\right\},\forall M\in{\bm Z}^+.
\end{equation}
\end{thm}
 \begin{proof}

Showing that $\tau_L(M) \geq \frac{1}{2}, \forall M \geq 1$ follows by a straightforward extension of the asymptotic interference alignment scheme of~\cite{Cadambe-IA}, and hence, it suffices to show that $\tau_L(M) \geq \frac{2M}{2M+L}$. 

We treat the network as a set of clusters, each consisting of consecutive $2M+L$ transceivers. The last $L$ transmitters of each cluster are deactivated to eliminate {\em inter-cluster} interference. It then suffices to show that $2M$ DoF can be achieved in each cluster. Without loss of generality, consider the cluster with users of indices in the set $[2M+L]$. We define the following subsets of $[2M+L]$,
\begin{eqnarray*}
{\cal S}_1 &=& [M],
\\{\cal S}_2 &=& \{L+M+1,L+M+2,\ldots,L+2M\}.
\end{eqnarray*}
We next show that each user in ${\cal S}_1 \cup {\cal S}_2$ achieves one degree of freedom, while messages $\{W_{M+1},W_{M+2},\ldots,W_{L+M}\}$ are not transmitted. In the proposed scheme, users in the set ${\cal S}_1$ are served by transmitters in the set $\{X_1,X_2,\ldots,X_M\}$ and users in the set ${\cal S}_2$ are served by transmitters in the set $\{X_{M+1},X_{M+2},\ldots,X_{2M}\}$. Let the message assignments be as follows.\\

${\cal T}_{i}=
\begin{cases}
\{i,i+1,\ldots,M\}, \quad &\forall i \in {\cal S}_1,\\
\{i-L,i-L-1,\ldots,M+1\},\quad &\forall  i \in {\cal S}_2.
\end{cases}$\\

Now, we note that messages with indices in ${\cal S}_1$ are not available outside transmitters with indices in $[M]$, and hence, do not cause interference at receivers with indices in ${\cal S}_2$. Also, messages with indices in ${\cal S}_2$ are not available at transmitters with indices in $[M]$, and hence, do not cause interference at receivers with indices in ${\cal S}_1$. 

In order to complete the proof by showing that each user in ${\cal S}_1 \cup {\cal S}_2$ achieves one degree of freedom, we next show that transmissions corresponding to messages with indices in ${\cal S}_1$(${\cal S}_2$) do not cause interference at receivers with indices in the same set. To avoid redundancy, we only describe in detail the design of transmit beams for message $W_1$ to cancel its interference at all receivers in ${\cal S}_1$ except its own receiver. First, the encoding of $W_1$ into $X_{1,1}$ at the first transmitter is done in a way that is oblivious to the existence of other receivers in the network except the first receiver, and a capacity achieving code for the point-to-point link $H_{1,1}$ is used. We then design $X_{2,1}$ at the second transmitter to cancel the interference caused by $W_1$ at the second receiver, i.e.,
\begin{equation}
X_{2,1}= -\frac{H_{2,1}}{H_{2,2}} X_{1,1}.
\end{equation}
Similarly, the transmit beam $X_{3,1}$ is then designed to cancel the interference caused by $W_1$ at the third receiver. The transmit beams $X_{i,1}, i\in\{2,3,\ldots,M\}$ are successively designed with respect to order of the index $i$ such that the received signal due to $X_{i,1}$ at the $i^{\mathrm{th}}$ receiver cancels the interference caused by $W_1$. 

In general, the availability of channel state information at the transmitters allows a design for the transmit beams for message $W_i$ that delivers it to the $i^{\mathrm{th}}$ receiver with a capacity achieving point-to-point code and simultaneously cancels its effect at receivers with indices in the set ${\cal C}_i$, where,\\

${\cal C}_{i}=
\begin{cases}
\{i+1, i+2, \ldots,M\},\quad  &\forall i \in {\cal S}_1,\\
\{i-1, i-2, \ldots,L+M+1\},\quad &\forall  i \in {\cal S}_2.
\end{cases}$\\

Note that both ${\cal C}_M$ and ${\cal C}_{L+M+1}$ equal the empty set, because both $W_M$ and $W_{L+M+1}$ do not contribute to interfering signals at receivers with indices in the set ${\cal S}_1 \cup {\cal S}_2$. %The above scheme for $M=3, L=1$ is illustrated in %Figure~\ref{fig:mthree} (a).
We conclude that each receiver with index in the set ${\cal S}_1\cup{\cal S}_2$ suffers only from Gaussian noise, thereby enjoying one degree of freedom.

\end{proof}

We refer the reader to Figure~\ref{fig:codingscheme} for an illustration of the above described coding scheme. We note that in the above coding scheme, some messages are not being transmitted in order to allow for interference-free communication for the remaining messages. It is worth noting that this can be done while maintaining fairness in the allocation of the available DoF over all users through \emph{fractional reuse} in a system where multiple sessions of communication take place, and different sets of receivers are deactivated in different sessions, e.g., in different time slots or different sub-carriers (in an OFDM system).

\begin{figure}
  \centering
\subfloat[]{\label{fig:mthreelone}\includegraphics[height=0.345\textwidth]{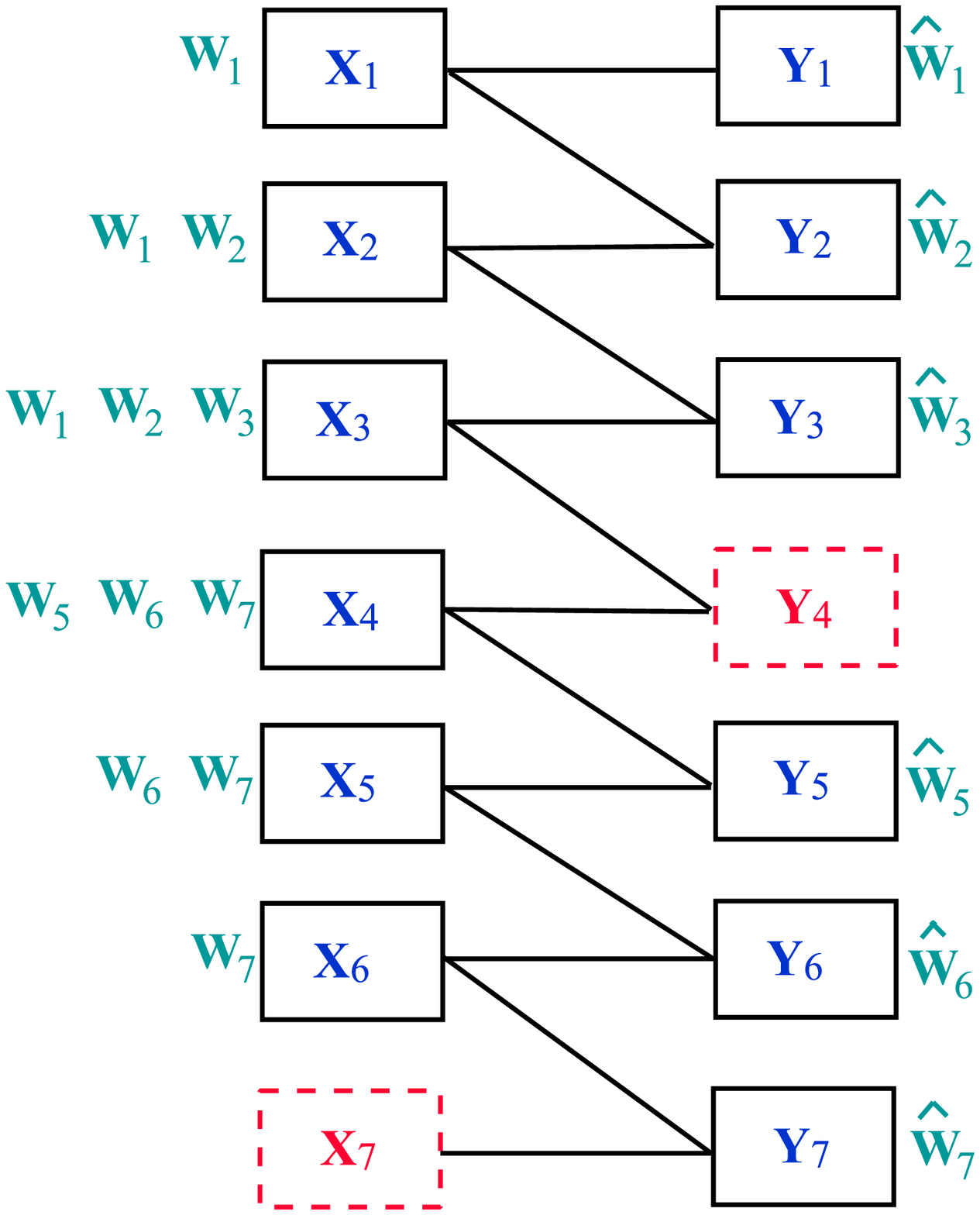}}                
\quad\quad\quad\quad\subfloat[]{\label{fig:mtwoltwo}\includegraphics[width=0.3\textwidth]{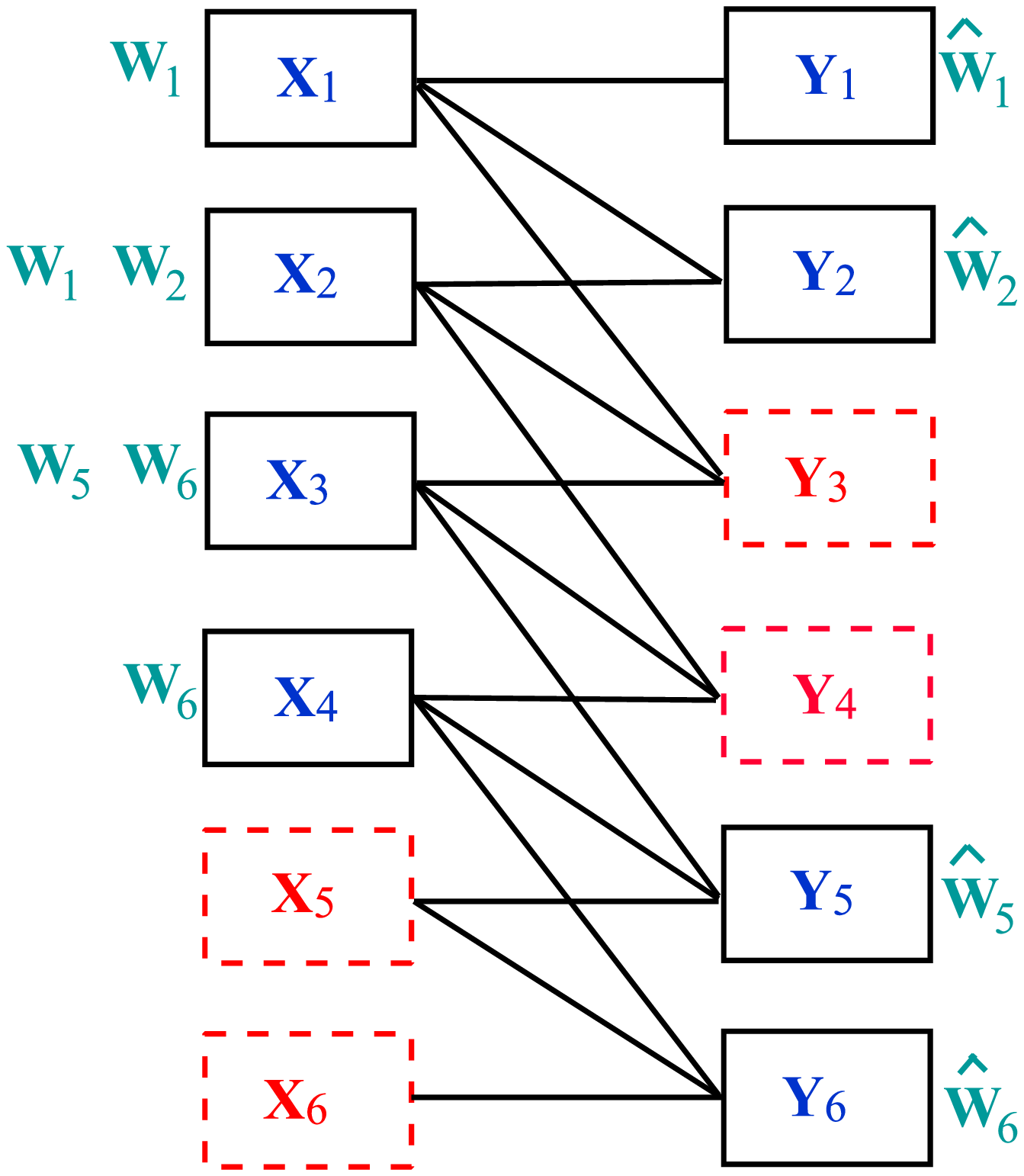}}
  \caption{Figure showing the assignment of messages in the proof of Theorem~\ref{thm:lowerbound} for the case where $M=3,L=1$ in ($a$), and the case where $M=L=2$ in ($b$). Only signals corresponding to the first cluster are shown. Signals in dashed red boxes are deactivated. Note that the last $L$ transmit signals are deactivated to eliminate inter-cluster interference. Also, messages $\{W_{M+1},\ldots,W_{M+L}\}$ are not transmitted, while each other message with indices in $\{1,2,\ldots,2M+L\}$ has one degree of freedom.}
  \label{fig:codingscheme}
\end{figure}

\subsection{Irreducible Message Assignments and Optimality of Local Cooperation}\label{sec:lcusefulmsgassignment}

In order to find an upper bound on the per user DoF $\tau_L(M)$, we have to consider all possible message assignment strategies satisfying the cooperation order constraint~\eqref{eq:coop_order}. In this section, we characterize necessary conditions for the optimal message assignment. The constraints we provide for transmit sets are governed by the connectivity pattern of the channel. For example, for the case where $M=1$, any assignment of message $W_i$ to a transmitter that is not connected to $Y_i$ is \emph{reducible}, i.e., the rate of transmitting message $W_i$ has to be zero for these assignments, and hence, removing $W_i$ from its carrying transmitter does not reduce the sum rate in these cases. 

We now introduce a graph theoretic representation that simplifies the presentation of the necessary conditions on irreducible message assignments. For message $W_i$, and a fixed transmit set ${\cal T}_i$, we construct the following graph $G_{W_i,{\cal T}_i}$ that has $[K]$ as its set of vertices, and an edge exists between any given pair of vertices $x,y \in [K]$ if and only if:
\begin{itemize}
\item  $x,y \in {\cal T}_i$,
\item  $|x-y| \leq L$.
\end{itemize}

Vertices corresponding to transmitters connected to $Y_i$ are given a special mark, i.e., vertices with labels in the set $\{i,i-1,\ldots,i-L\}$ are marked for the considered channel model. 

We now have the following statement.

%In Figure~\ref{fig:usefulmsgassignment}, we give an example for the construction of $G_{W_i,{\cal T}_i}$ and the application of the following lemma.

\begin{lem}\label{lem:usefulmsgdistribution}
For any $k \in {\cal T}_i$ such that the vertex $k$ in $G_{W_i,{\cal T}_i}$ is not connected to a marked vertex, removing $k$ from ${\cal T}_i$ does not decrease the sum rate.
\end{lem}

\begin{proof}
Let ${\cal S}$ denote the set of indices of vertices in a component with no marked vertices. We need to show that removing any transmitter index in ${\cal S}$ from ${\cal T}_i$ does not decrease the sum rate. Let ${\cal S}'$ be the set of indices of received signals that are connected to at least one transmitter with an index in ${\cal S}$. To prove the lemma, we consider two scenarios, where we add a {\em tilde} over symbols denoting signals belonging to the second scenario. For the first scenario, $W_i$ is made available at all transmitters with indices in ${\cal S}$. Let $Q$ be a random variable that is independent of all messages and has the same distribution as $W_i$, then for the second scenario, $W_i$ is not available at any transmitter with an index in ${\cal S}$, and a realization $q$ of $Q$ is generated and given to all transmitters with indices in ${\cal S}$ before communication starts. Moreover, the given realization $Q=q$ contributes to the encoding of $\tilde{X}_{\cal S}$ in the same fashion as a message $W_i=q$ contributes to $X_{\cal S}$. Assuming a reliable communication scheme for the first scenario that uses a large block length $n$, the following argument shows that the achievable sum rate is also achievable after removing $W_i$ from the designated transmitters. And therefore, proving that removing any transmitter in ${\cal S}$ from ${\cal T}_i$ does not decrease the sum rate.
\begin{eqnarray*}
n\sum_j R_j &=& \sum_j \textsf{H}(W_j) \notag
\\&\overset{(a)}{\leq}& \sum_j I(W_j;Y_j^n) + o(n) \notag
\\&=& \sum_{j \in {\cal S}'^c} I(W_j,Y_j^n) + \sum_{j \in {\cal S}'} I(W_j;Y_j^n)+ o(n) \notag
\\&\overset{(b)}{=}& \sum_{j \in {\cal S}'^c} I(W_j,\tilde{Y}_j^n) + \sum_{j \in {\cal S}'} I(W_j;Y_j^n)+ o(n) \notag
\\&\overset{(c)}{=}&\sum_{j \in {\cal S}'^c} I(W_j,\tilde{Y}_j^n) + \sum_{j \in {\cal S}'} I(W_j;\tilde{Y}_j^n)+ o(n) \notag
\\&=& \sum_j \textsf{H}(W_j) - \textsf{H}(W_j|\tilde{Y}_j^n) + o(n),
\end{eqnarray*}
where $\textsf{H}(.)$ is the entropy function for discrete random variables, $(a)$ follows from Fano's inequality, $(b)$ follows as the difference between the two scenarios lies in the encoding of $X_{\cal S}$ which affects only $Y_{{\cal S}'}$, and $(c)$ holds because any two transmitters carrying $W_i$ and connected to a receiver whose index is in ${\cal S'}$ must belong to the same component, and hence, $W_i$ contributes to $Y_{\cal S'}$ only through $X_{\cal S}$, it follows that $(W_j,Y_j^n)$ has the same joint distribution as $(W_j,\tilde{Y}_j^n)$ for every $j \in S'$. Now, it follows that,
\begin{equation}
\sum_j \textsf{H}(W_j|\tilde{Y}_j^n) = o(n),
\end{equation}
and hence, the rates $R_j, j\in[K]$ are achievable in the second scenario.
\end{proof}
\begin{figure}
\centering
\includegraphics[width=0.5\columnwidth]{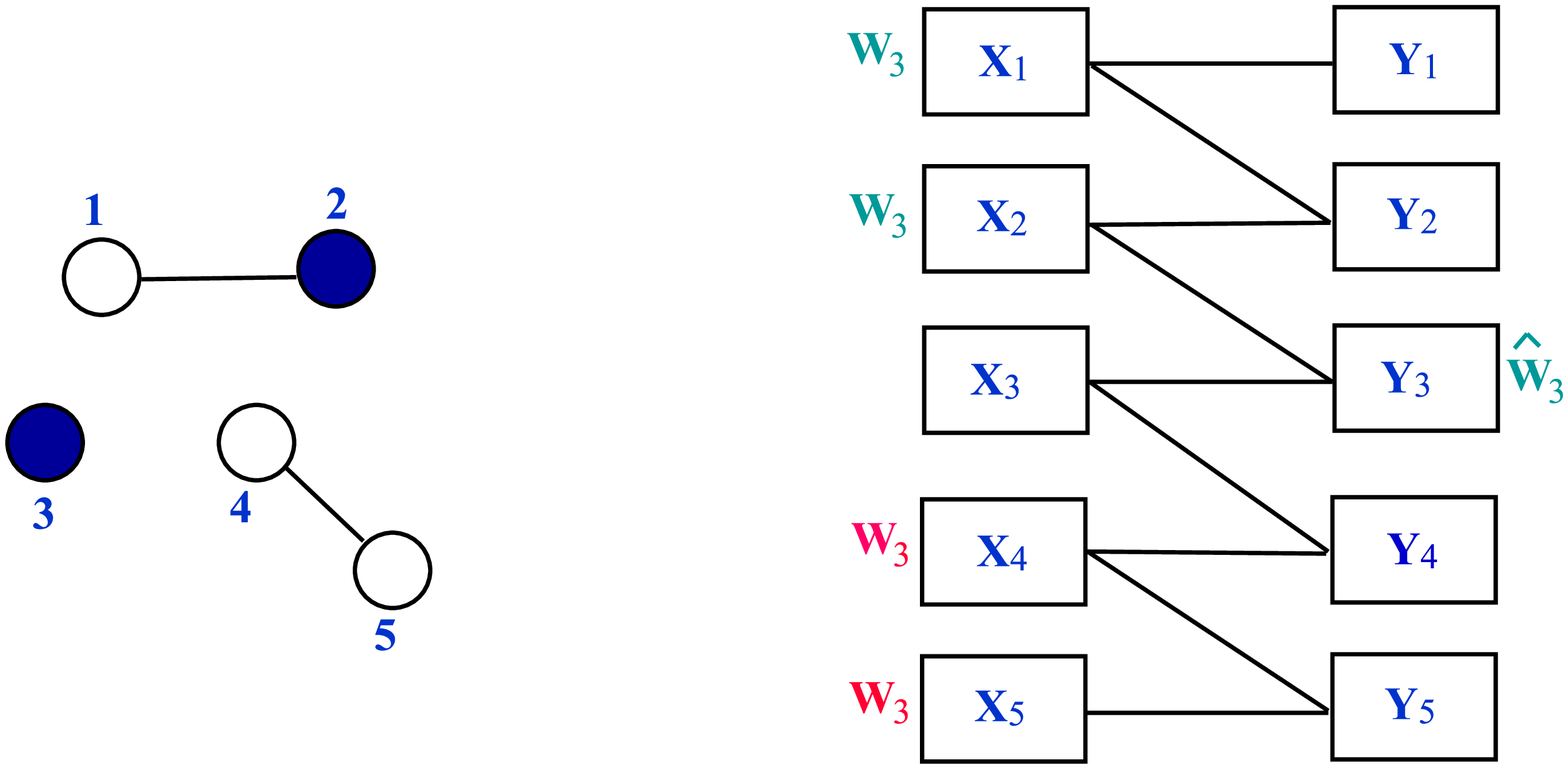}
\caption{Figure showing the construction of $G_{W_3,{\cal T}_3}$ in a $5-$user channel with $L=1$. Marked vertices are represented with filled circles. $W_3$ can be removed at both $X_4$ and $X_5$ without decreasing the sum rate, as the corresponding vertices lie in a component that does not contain a marked vertex.}
\label{fig:usefulmsgassignment}
\end{figure}
We call a message assignment \emph{irreducible} if no element in it can be removed without decreasing the sum rate. The following corollary to the above lemma characterizes a necessary condition for any message assignment satisfying the cooperation order constraint in~\eqref{eq:coop_order} to be irreducible. Recall that two vertices in a graph $G$ are at a distance $d$ if and only if the shortest path in $G$ between the two vertices has $d$ edges. 
\begin{cor}\label{cor:usefulmsgdistribution}
Let ${\cal T}_i$ be an irreducible message assignment and $|{\cal T}_i| \leq M$, then $\forall k\in[K], k\in{\cal T}_i$ only if the vertex $k$ in $G_{W_i,{\cal T}_i}$ lies at a distance that is less than or equal $M-1$ from a marked vertex.
\end{cor}
Note that in the considered channel model, the above result implies that ${\cal T}_i \subseteq \{i-ML,i-ML+1,\ldots,i+(M-1)L\}$, from which we obtain the following result.
\begin{thm}
Local cooperation is optimal for locally connected channels, 
\begin{equation}
\tau_L^{(\textrm{loc})}(M) = \tau_L(M), \forall M,L\in{\bm Z}^{+}.
\end{equation} 
\end{thm}

And so we note that even though local cooperation does not achieve a scalable DoF gain for the fully connected channel, not only does it achieve a scalable gain when the connectivity assumption is relaxed to local connectivity, but the confinement to local cooperation no longer results in a loss in the available DoF.

\subsection{DoF Upper Bounds}\label{sec:lcupperbound}
In this section, we prove upper bounds on $\tau_1(M)$ and $\tau_L(1)$ that establishes the tightness of the lower bound in Theorem~\ref{thm:lowerbound} for the special cases where either $L=1$ or $M=1$. First, in order to assess the optimality of the coding scheme introduced in Section~\ref{sec:dofgains} for arbitrary values of the system parameters, we prove a general upper bound for a class of coding schemes that only employs a zero-forcing transmit beam-forming strategy.

\subsubsection{ZF Transmit Beam-Forming}
 Consider only coding schemes with transmit signals of the form~\eqref{eq:linear-precoding} and each message is either not transmitted or allocated one degree of freedom. More precisely, let $\tilde{Y}_j=Y_j-Z_j, \forall j\in[K]$, then in addition to the constraint in~\eqref{eq:linear-precoding}, it is either case that the mutual information $I(\tilde{Y}_j;W_j)=0$ or it is the case that $W_j$ completely determines $\tilde{Y}_j$. Note that $\tilde{Y}_j$ can be determined from $W_j$ for the case where user $j$ enjoys interference-free communication and $I(W_j;\tilde{Y}_j)=0$ for the other case where $W_j$ is not transmitted. We say that the $j^{\mathrm{th}}$ receiver is \emph{active} if and only if $I(\tilde{Y}_j;W_j)>0$. Note that using zero-forcing transmit beamforming, if the $j^{\mathrm{th}}$ receiver is active, then $I(W_i;Y_j)=0, \forall i \neq j$. 

Let $\tau_{L}^{(\textrm{zf})}(M)$ denote the asymptotic characterization of the per user DoF under the restriction to the above described class of coding schemes. In Theorem~\ref{thm:zfupperbound} below, we show that the coding scheme in the proof of Theorem~\ref{thm:lowerbound} achieves the optimal value of $\tau_{L}^{(\textrm{zf})}(M)$. We first prove Lemma~\ref{lem:lemone} that bounds the number of receivers at which the interference of a given message can be cancelled. 

For a set ${\cal S}\subseteq [K]$, let ${\cal V}_{\cal S}$ be the set of indices for active receivers connected to transmitters with indices in ${\cal S}$. More precisely, ${\cal V}_{\cal S} = \{j: I(\tilde{Y}_j;W_j)>0, {\cal S} \cap \{j,j-1,\ldots,j-L\} \neq \phi \}$, where $\phi$ is the empty set. To obtain the following results, we assume that for each transmitter in ${\cal T}_i$, message $W_i$ contributes to the transmit signal of  this transmitter. i.e., $\forall j \in {\cal T}_i, I(W_i,X_j) > 0$. Note that this assumption does not introduce a loss in generality, because otherwise the transmitter can be removed from ${\cal T}_i$. We need the following lemma for the proof of the upper bound on $\tau_{L}^{(\textrm{zf})}(M)$ in Theorem~\ref{thm:zfupperbound}.

\begin{lem}\label{lem:lemone}
For any message $W_i$, the number of active receivers connected to at least one transmitter carrying the message is no greater than the number of transmitters carrying the message.
\begin{equation}\label{eq:lemoneproof}
|{\cal V}_{{\cal T}_i}| \leq |{\cal T}_i|.
\end{equation}
\begin{proof}
We only consider the non-trivial case where ${\cal T}_i \neq \phi$. For each receiver $j \in{\cal V}_{{\cal T}_i}$, there exists a transmit signal $X_{k,i}$, $k\in[K]$ such that conditioned on all other transmit signals, the received signal $Y_j$ is correlated with the message $W_i$. More precisely, $I\left(W_i;Y_j|\{X_{v,i}, v\in[K], v\neq k\}\right)>0$. Now, since we impose the constraint $I(W_i;Y_j)=0, \forall j\in{\cal V}_{{\cal T}_i}$, the interference seen at all receivers in ${\cal V}_{{\cal T}_i}$ has to be cancelled. Finally, since the probability of a zero Lebesgue measure set of channel realizations is zero, the $|{\cal T}_i|$ transmit signals carrying $W_i$ cannot be designed to cancel $W_i$ at more than $|{\cal T}_i|-1$ receivers for almost all channel realizations.   
\end{proof}
\end{lem}

\begin{thm}\label{thm:zfupperbound}
Under the restriction to ZF Transmit Beam-Forming coding schemes (interference avoidance), the asymptotic per user DoF of a locally connected channel with connectivity parameter $L$ is given by,
\begin{equation}\label{eq:zfupperbound}
\tau_{L}^{(\textrm{zf})}(M) = \frac{2M}{2M+L}.
\end{equation}
\end{thm}
\begin{proof}
The proof of the lower bound is the same as the proof of Theorem~\ref{thm:lowerbound} for the case where $\frac{2M}{2M+L} > \frac{1}{2}$. It then suffices to show that $\tau_{L}^{(\textrm{zf})}(M) \leq \frac{2M}{2M+L}$. 

In order to prove the upper bound, we show that the sum degree of freedom in each set ${\cal S} \subseteq [K]$ of consecutive $2M+L$ users is bounded by $2M$. We now focus on proving this statement by fixing a set ${\cal S}$ of consecutive $2M+L$ users, and make the following definitions. For a user $i\in[{\cal S}]$, let ${\cal U}_i$ be the set of active users in ${\cal S}$ with an index $j > i$, i.e., 
\begin{equation*}
{\cal U}_i=\{j: j>i, j\in{\cal S}, I(\tilde{Y}_j;W_j)>0\}.
\end{equation*}
Similarly, let ${\cal D}_i$ be the set of active users in ${\cal S}$ with an index $j < i$,
\begin{equation*}
{\cal D}_i=\{j: j<i, j\in{\cal S}, I(\tilde{Y}_j;W_j)>0\}.
\end{equation*}
Assume that ${\cal S}$ has at least $2M+1$ active users, then there is an active user in ${\cal S}$ that lies in the middle of a subset of $2M+1$ active users in ${\cal S}$.
More precisely, $\exists i \in {\cal S} : |{\cal T}_i|>0, |{\cal U}_i|\geq M, |{\cal D}_i| \geq M$, we let this middle user have the $i^{\mathrm{th}}$ index for the rest of the proof.

Let $s_{min}$ and $s_{max}$ be the users in ${\cal S}$ with minimum and maximum indices, respectively, i.e., $s_{min} = \min_s \{s: s \in {\cal S}\}$ and $s_{max} = \max_s \{s: s \in {\cal S}\}$, we then consider the following cases to complete the proof,

{\bf Case 1:} $W_i$ is being transmitted from a transmitter that is connected to the receiver with index $s_{min}$, i.e., $\exists s \in {\cal T}_i: s \in \{s_{min},s_{min}-1,\ldots,s_{min}-L\}$. It follows from Lemma~\ref{lem:usefulmsgdistribution} that ${\cal V}_{{\cal T}_i} \supseteq {\cal D}_i \cup \{i\}$, and hence, $|{\cal V}_{{\cal T}_i}| \geq M+1$, which contradicts~\eqref{eq:lemoneproof}, as $|{\cal T}_i|\leq M$. 

{\bf Case 2:} $W_i$ is being transmitted from a transmitter that is connected to the receiver with index $s_{max}$, i.e., $\exists s \in {\cal T}_i: s \in \{s_{max},s_{max}-1,\ldots,s_{max}-L\}$. It follows from Lemma~\ref{lem:usefulmsgdistribution} that ${\cal V}_{{\cal T}_i} \supseteq {\cal U}_i \cup \{i\}$, and hence, $|{\cal V}_{{\cal T}_i}| \geq M+1$, which again contradicts~\eqref{eq:lemoneproof}. 

{\bf Case 3:} For the remaining case, there is no transmitter in ${\cal T}_i$ that is connected to any of the receivers with indices $s_{min}$ and $s_{max}$. In this case, it follows from Lemma~\ref{lem:usefulmsgdistribution} that ${\cal T}_i$ does not contain a transmitter that is connected to a receiver with an index less than $s_{min}$ or greater than $s_{max}$, and hence, all the receivers connected to transmitters carrying $W_i$ belong to ${\cal S}$. It follows that at least $L+|{\cal T}_i|$ receivers in ${\cal S}$ are connected to one or more transmitter in ${\cal T}_i$, and since ${\cal S}$ has at least $2M+1$ active receivers, then any subset of $L+|{\cal T}_i|$ receivers in ${\cal S}$ has to have at least $2M+1-((2M+L)-(L+|{\cal T}_i|))=|{\cal T}_i|+1$ active receivers, and the statement is proved by reaching a contradiction to~\eqref{eq:lemoneproof} in the last case.
\end{proof}

\subsubsection{Wyner's Asymmetric Model}
Now, we consider the special case of $L=1$, and prove that the lower bound stated in Theorem~\ref{thm:lowerbound} is tight in this case. We start by stating the following auxiliary lemma for any $K$-user Gaussian interference channel with a DoF number of $\eta$. For any set of receiver indices ${\cal A} \subseteq [K]$, define $U_{\cal A}$ as the set of indices of transmitters that exclusively carry the messages for the receivers in ${\cal A}$, and the complement set $\bar{U}_{\cal A}$ is the set of indices of transmitters that carry messages for receivers outside ${\cal A}$. More precisely, $U_{\cal A} = [K]\backslash\cup_{i \notin {\cal A}} {\cal T}_i$, then, 
\begin{lem}\label{lem:dofouterbound}
If there exists a set ${\cal A}\subseteq [K]$, a function $f_1$, and a function $f_2$ whose definition does not depend on the transmit power constraint $P$, and $f_1\left(Y_{\cal A},X_{U_{\cal A}}\right)=X_{\bar{U}_{\cal A}}+f_2(Z_{\cal A})$, then $\eta \leq |{\cal A}|$. 
\end{lem} 
\begin{proof}
The proof is available in the Appendix. Here, we provide a sketch. Recall that $Y_{\cal A}=\{Y_i, i\in {\cal A}\}$, and $W_{\cal A}=\{W_i, i\in {\cal A}\}$, and note that $X_{U_{\cal A}}$ is the set of transmit signals that do not carry messages outside $W_{\cal A}$. Fix a reliable communication scheme for the considered $K-$user channel, and assume that there is only one centralized decoder that has access to the received signals $Y_{\cal A}$. We show that using the centralized decoder, the only uncertainty in recovering all the messages $W_{[K]}$ is due to the Gaussian noise signals. In this case, the sum DoF is bounded by $|{\cal A}|$, as it is the number of received signals used for decoding. 

Using $Y_{\cal A}$, the messages $W_{\cal A}$ can be recovered reliably, and hence, the signals $X_{\bar{U}_{\cal A}}$ can be reconstructed. Using $Y_{\cal A}$ and $X_{\bar{U}_{\cal A}}$, the remaining transmit signals can be approximately reconstructed using the function $f_1$ of the hypothesis. Finally, using all transmit signals, the received signals $Y_{\bar{\cal A}}$ can be approximately reconstructed, and the messages $W_{\bar{\cal A}}$ can then be recovered.  
\end{proof}

We note that Lemma~\ref{lem:dofouterbound} applies to all considered channel models. Now, we use it to prove a DoF upper bound for Wyner's model.

\begin{thm}\label{thm:Asymmetric Model}
The asymptotic per user DoF for Wyner's asymmetric model with CoMP transmission is given by,
\begin{equation}
\tau_1(M)=\frac{2M}{2M+1}, \forall M \in {\bm Z}^+.
\end{equation}
\end{thm}
\begin{proof}
The lower bound follows from Theorem~\ref{thm:lowerbound}. In order to prove the converse, we use Lemma~\ref{lem:dofouterbound} with a set ${\cal A}$ of size $K\frac{2M}{2M+1}+o(K)$. We also prove the upper bound for the channel after removing the first $M$ transmitters $\left(X_{[M]}\right)$, while noting that this will be a valid bound on $\tau_1(M)$ since the number of removed transmitters is $o(K)$. 

Inspired by the coding scheme in the proof of Theorem~\ref{thm:lowerbound}, we define the set ${\cal A}$ as the set of receivers that are {\em active} in the coding scheme. i.e., the complement set ${\bar{\cal A}}=\{i: i\in[K], i= (2M+1)(j-1)+M+1, j \in {\bm{Z}^+}\}$. We know from Corollary~\ref{cor:usefulmsgdistribution} that messages belonging to the set $W_{\bar{\cal A}}$ do not contribute to transmit signals with indices that are multiples of $2M+1$, i.e., $i \in U_{\cal A}$ for all $i\in[K]$ that is a multiple of $2M+1$. More precisely, let the set ${\cal S}$ be defined as follows:
\begin{equation*}
{\cal S} = \{i: i\in[K], i \text{ is a multiple of } 2M+1\},
\end{equation*}
then ${\cal S}\subseteq {U_{\cal A}}$. In particular, $X_{\cal S} \subseteq X_{U_{\cal A}}$, and hence it suffices using Lemma~\ref{lem:dofouterbound} to show the existence of linear functions $f_1$ and $f_2$ such that $f_1\left(Y_{\cal A},X_{\cal S}\right)=X_{\bar{S}}\backslash X_{[M]}+f_2(Z_{\cal A})$. 

In what follows we show how to reconstruct a noisy version of the signals in the set 
$\{X_{M+1},X_{M+2},\ldots,X_{2M}\}$ $\cup$
$\{X_{2M+2},X_{2M+3},\ldots,X_{3M+1}\}$, where the reconstruction noise depends only on $Z_{\cal A}$ in a linear fashion. Then it will be clear by symmetry how to reconstruct the rest of transmit signals in the set $X_{\bar{\cal S}}\backslash X_{[M]}$. Since $X_{2M+1} \in X_{\cal S}$ and ${Y_{2M+1}}$ is also given, $X_{2M}+Z_{2M+1}$ can be reconstructed. Now, with the knowledge of $X_{2M}+Z_{2M+1}$ and $Y_{2M}$, we can reconstruct $X_{2M-1}+Z_{2M}-Z_{2M+1}$, and so by iterative processing, a noisy version of all transmit signals in the set $\{X_{M+1},X_{M+2},\ldots,X_{2M}\}$ can be reconstructed, where the noise is a linear function of the signals $\{Z_{M+2},Z_{M+3},\ldots,Z_{2M+1}\}$. In a similar fashion, given $X_{2M+1}$ and $Y_{2M+2}$, the signal $X_{2M+2}+Z_{2M+2}$ can be reconstructed. Then with the knowledge of $Y_{2M+3}$, we can reconstruct $X_{2M+3}+Z_{2M+3}-Z_{2M+2}$, and we can proceed along this path to reconstruct a noisy version of all transmit signals in the set $\{X_{2M+2},X_{2M+3},\ldots,X_{3M+1}\}$, where the noise is a linear function of the signals $\{Z_{2M+2},Z_{2M+3},\ldots,Z_{3M+1}\}$. This proves the existence of linear functions $f_1$ and $f_2$ such that $f_1(Y_1,X_{\cal S})=X_{\bar{S}}\backslash X_{[M]}+f_2(Z_{\cal A})$, and the coefficients for $f_2$ do not depend on the transmit power constraint $P$, and so by Lemma~\ref{lem:dofouterbound} we obtain the converse of Theorem~\ref{thm:Asymmetric Model}.
\end{proof}

In Figure~\ref{fig:mthree} (b), we illustrate how the proof works for the case where $M=3$. Note that the missing received signals $\{Y_4,Y_{11},\ldots\}$ in the upper bound proof correspond to the inactive receivers in the coding scheme.

\begin{figure}
  \centering 
\subfloat[]{\label{fig:mthreejonecs}\includegraphics[height=0.283\textwidth]{M3J1CodingSchemeV2.eps}}                
\quad\quad\quad\quad\subfloat[]{\label{fig:mthreejoneub}\includegraphics[width=0.17\textwidth]{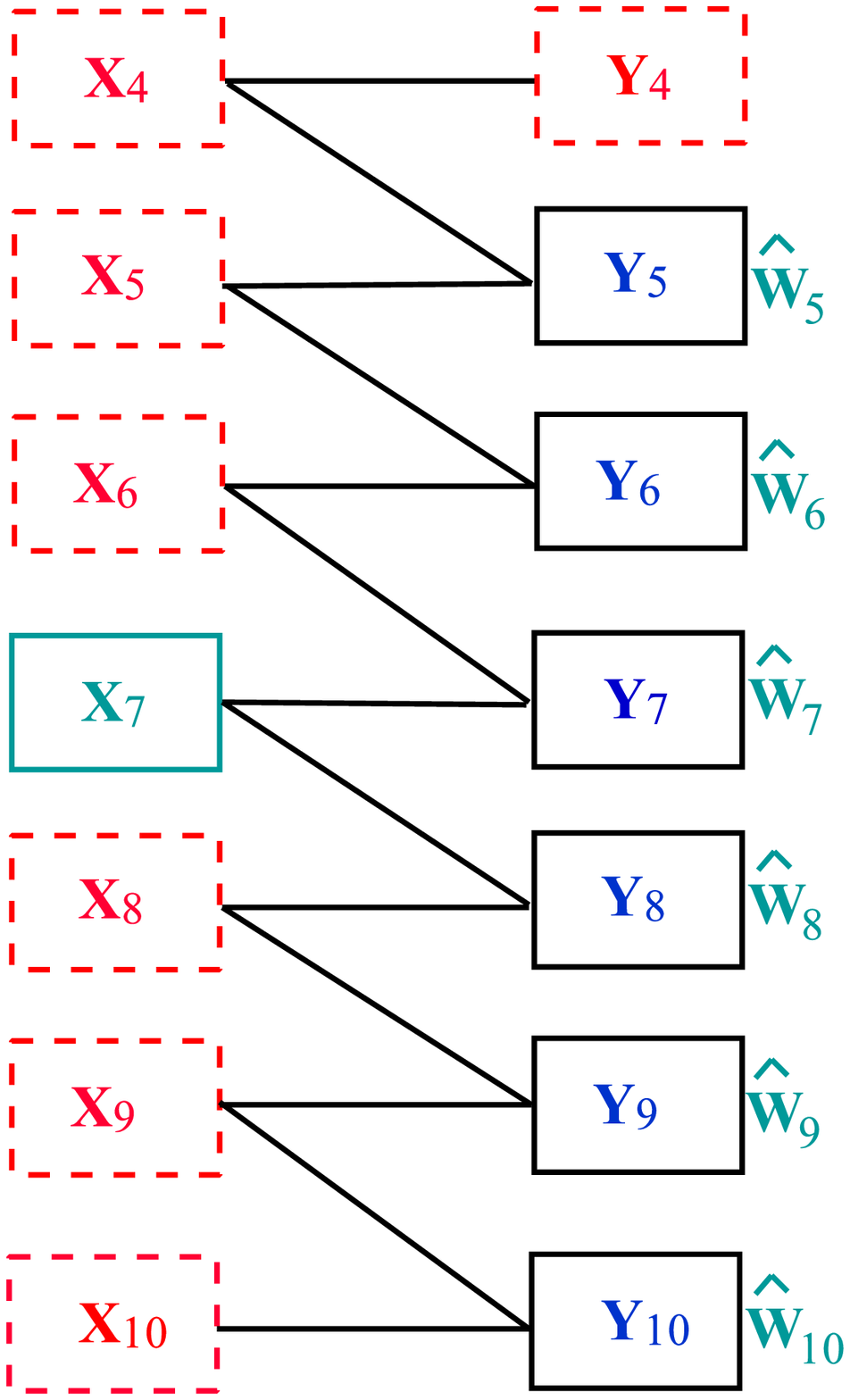}}
  \caption{Figure illustrating the proof of Theorem~\ref{thm:Asymmetric Model} for $M=3$, $\tau(3)=\frac{6}{7}$. In ($a$), the message assignments in the first cluster for the proposed coding scheme are illustrated. Note that both $X_7$ and $Y_4$ are deactivated. In ($b$), an illustration of the upper bound is shown. The messages $W_4$ and $W_{11}$ cannot be available at $X_7$, hence it can be reconstructed from $W_{\cal A}$. A noisy version of all transmit signals shown in figure can be reconstructed from $X_7$ and the signals $\{Y_5,\ldots,Y_{10}\}$, where the reconstruction noise is a linear function of $\{Z_5,\ldots,Z_{10}\}$.}
  \label{fig:mthree}
\end{figure}

\subsubsection{No Cooperation}
We note that even for the case of {\em no cooperation}, an asymptotic per user DoF of more than $\frac{1}{2}$ per user DoF is achievable, i.e., $\tau_1(1) = \frac{2}{3}$. Also, it is straightforward to see that the interference alignment scheme can be generalized to show that $\tau_L(1) \geq \frac{1}{2}$ for any locally connected channel with parameter $L$. The next theorem generalizes the upper bound in~\cite{Madsen-Nosratinia} for locally connected channels, where each message can be available at one transmitter that is not necessarily the transmitter carrying its own index. In particular, we show that $\tau_L(1) > \frac{1}{2}$ only if $L = 1$. 

The following lemma serves as a building block for the upper bound proof in Theorem~\ref{thm:nocooperation} below. We define ${\cal R}_i$ as the set of indices of received signals that are connected to transmitter $X_i$, i.e., ${\cal R}_i = \{i,i+1,\ldots,i+L\}$. Note that as we are considering the case of {\em no cooperation}, hence, ${\cal T}_i$ contains only one element. Recall that $d_i$ denotes the available DoF for the communication of message $W_i$.

\begin{lem}\label{lem:cornerstone}
If ${\cal T}_i = \{X_j\}$, then $d_i + d_s \leq 1, \forall s \in {\cal R}_j, s \neq i$.
\end{lem}
\begin{proof}
We assume that all messages other than $W_i$ and $W_s$ are deterministic, and then apply Lemma~\ref{lem:dofouterbound} with the set ${\cal A}=\{s\}$, and functions $f_1$ and $f_2$ defined such that the following holds,
\begin{eqnarray*}
f_1\left(Y_s,X_{[K]\backslash\{j\}}\right)&=& H_{s,j}^{-1}\left(Y_s-{\bm H}_{\{s\},[K]\backslash\{j\}}X_{[K]\backslash\{j\}}\right)
\\&=& X_j + H_{s,j}^{-1}Z_s
\\&=& X_j + f_2(Z_s),
\end{eqnarray*}
and then the bound follows.
\end{proof}

\begin{thm}\label{thm:nocooperation}
Without cooperation ($M=1$), the asymptotic per user DoF of locally connected channels is given by,

$\tau_L(1)=
\begin{cases}
\frac{2}{3},\quad  &\text{if} \quad L=1,\\
\frac{1}{2},\quad &\text{if} \quad L \geq 2.
\end{cases}$

\end{thm}

\vspace{5 mm}
\begin{proof}
The case where $L=1$ is a special case of the result in Theorem~\ref{thm:Asymmetric Model}. The lower bound for the case where $L \geq 2$ follows by assigning each message to the transmitter with the same index, and a simple extension of the asymptotic interference alignment scheme of~\cite{Cadambe-IA}, and hence, it suffices to show that,
\begin{equation}
\tau_L(1) \leq \frac{1}{2}, \forall L\geq 2.
\end{equation}
In this proof, we use the locally connected channel model defined in~\eqref{eq:channel}. Each transmitter is connected to $\left \lfloor \frac{L}{2} \right \rfloor$ preceding receivers and $\left \lceil \frac{L}{2} \right \rceil$ succeeding receivers. In order to prove the theorem statement, we establish the stronger statement,
\begin{equation}\label{eq:stronger}
\eta_L(K,1) \leq \frac{K+1}{2}, \forall K, \forall L \geq 2.
\end{equation}
We prove~\eqref{eq:stronger} by induction. The basis to the induction step is given by the following.
\begin{equation}\label{eq:basisnocoop}
\text{For }M=1, \forall L, d_1 + d_2 \leq 1.
\end{equation}
The proof of~\eqref{eq:basisnocoop} follows from Lemma~\ref{lem:cornerstone} and the fact that all transmitters connected to $Y_1$ are also connected to $Y_2$. In order to state the induction step, we first define $B_k$ as a boolean variable that is true if and only if the following is true:
\begin{itemize}
\item $\frac{\sum_{i=1}^{k} d_i}{k} \leq \frac{1}{2}$.
\item $d_{k-1} + d_k \leq 1$.
\end{itemize}
The induction step is given by the following.
\begin{equation}\label{eq:inductionstepnocoop}
\text{For }L\geq 2, k \geq 2, \text{ if } B_k \text{ is true}, \text{ then either } B_{k+1} \text{ or } B_{k+2} \text{ is true.}
\end{equation}
In order to prove~\eqref{eq:inductionstepnocoop}, consider the assignment of message $W_{k+1}$, and note that $W_{k+1}$ is available at a transmitter connected to $Y_{k+1}$. Now, note that $\forall L \geq 2$, the channel model of~\eqref{eq:channel} implies that any transmitter connected to $Y_{k+1}$ is either connected to $Y_{k+2}$, or to both $Y_k$ and $Y_{k-1}$. The proof follows by considering these two cases separately.

{\bf Case 1:} If $W_{k+1}$ is available at a transmitter that is connected to $Y_{k+2}$, then it follows from Lemma~\ref{lem:cornerstone} that $d_{k+1}+d_{k+2} \leq 1$. Since $B_k$ is true, it follows that $\frac{\sum_{i=1}^{k} d_i}{k} \leq \frac{1}{2}$, and hence, $\frac{\sum_{i=1}^{k+2} d_i}{k+2} \leq \frac{1}{2}$. In this case,~\eqref{eq:inductionstepnocoop} holds since $B_{k+2}$ is true.

{\bf Case 2:} If $W_{k+1}$ is available at a transmitter that is connected to both $Y_k$ and $Y_{k-1}$, then it follows from Lemma~\ref{lem:cornerstone} that $d_{k+1}+d_k \leq 1$, and $d_{k+1} + d_{k-1} \leq 1$. Now, since $B_k$ is true, it follows that $d_k+d_{k-1} \leq 1$, and hence,
\begin{equation}\label{eq:proofone}
\frac{d_{k+1}+d_k+d_{k-1}}{3} \leq \frac{1}{2}.
\end{equation}
Also, since $B_k$ is true, we know that $\frac{\sum_{i=1}^{k-2} d_i}{k-2} \leq \frac{1}{2}$, and hence, we get from~\eqref{eq:proofone} that $\frac{\sum_{i=1}^{k+1} d_i}{k+1} \leq \frac{1}{2}$. In this case,~\eqref{eq:inductionstepnocoop} holds since $B_{k+1}$ is true. 

It follows by induction from~\eqref{eq:basisnocoop} and~\eqref{eq:inductionstepnocoop} that it is either the case that $B_{K-1}$ is true, or $B_{K}$ is true. If $B_{K-1}$ is true, then $\sum_{i=1}^{K-1} d_i \leq \frac{K-1}{2}$, and the DoF number $\eta \leq \frac{K+1}{2}$. If $B_K$ is true, then it follows that the dof number $\eta \leq \frac{K}{2}$.

\end{proof}

\section{Discussion}\label{sec:discussion}

There are two design parameters in the considered problem, the message assignment strategy satisfying the cooperation order constraint, and the design of transmit beams. We characterized the asymptotic per user DoF when one of the design parameters is restricted to a special choice, i.e., restricting message assignment strategies by a local cooperation constraint or restricting the design of transmit beams to zero-forcing transmit beams. The restriction of one of the design parameters can significantly simplify the problem because of the inter-dependence of the two design parameters.  On one hand, the achievable scheme is enabled by the choice of the message assignment strategy, and on the other hand, the assignment of messages to transmitters is governed by the technique followed in the design of transmit beams, e.g. zero-forcing transmit beamforming or interference alignment. In the following, we discuss each of the design parameters.

\subsection{Message Assignment Strategy}

The assignment of each message to more than one transmitter (CoMP transmission) creates a virtual Multiple Input Single Output (MISO) network. A real MISO network, where multiple dedicated antennas are assigned to the transmission of each message, differs from the created virtual one in two aspects. First, in a CoMP transmission setting, the same transmit antenna can carry more than one message. Second, for locally connected channels, the number of receivers at which a message causes undesired interference depends on the number of transmit antennas carrying the message. 

For fully connected channels, the number of receivers at which a message causes undesired interference is the same regardless of the size of the transmit set as long as it is non-empty. The only aspect that governs the assignment of messages to transmitters is the pattern of overlap between transmit sets corresponding to different messages. It is expected that the larger the sizes of the intersections between sets of messages carried by different transmit antennas, the more dependent the coefficients of the virtual MISO channel are, and hence, the lower the available DoF. For the spiral assignments of messages considered in~\cite{Annapureddy-ElGamal-Veeravalli-IT11}, $|{\cal T}_i \cap {\cal T}_{i+1}| = M-1$, and the same value holds for the size of the intersection between sets of messages carried by successive transmitters. In general, local cooperation implies large intersections between sets of messages carried by different transmitters, and hence, the negative conclusion we reached for $\tau^{(\textrm{loc})}(M)$.  

For the case where we are restricted to zero-forcing transmit beamforming as in Section~\ref{sec:dofgains}, the number of receivers at which each message causes undesired interference governs the choice of transmit sets, and hence, we saw that for locally connected channels, the message assignment strategy illustrated in Theorem~\ref{thm:lowerbound} selects transmit sets that consist of successive transmitters, to minimize the number of receivers at which each message should be cancelled. This strategy is optimal under the restriction to zero-forcing transmit beamforming schemes.

\subsection{Design of Transmit Beams}

While it was shown in~\cite{Annapureddy-ElGamal-Veeravalli-IT11} that CoMP transmission accompanied by both zero-forcing transmit beams and asymptotic interference alignment can achieve a DoF cooperation gain beyond what can be achieved using only transmit zero-forcing, this is not obvious for locally connected channels. Unlike in the fully connected channel, the addition of a transmitter to a transmit set in a locally connected channel may result in an increase in the number of receivers at which the message causes undesired interference. 

We note that unlike asymptotic interference alignment scheme, the zero-forcing transmit beamforming scheme illustrated in Section~\ref{sec:dofgains} does not need symbol extensions, since it achieves the stated DoF of Theorem~\ref{thm:lowerbound} in one channel realization. 
However, it is not clear whether asymptotic interference alignment can be used to show an asymptotic per user DoF cooperation gain beyond that achieved through simple zero-forcing transmit beamforming; we believe that the answer to this question is closely related to both problems that remain open after this work, i.e., characterizing $\tau(M)$ and $\tau_L(M)$. 

\section{Conclusions}\label{sec:conclusions}

We studied the DoF gain achieved through CoMP transmission. In particular, it was of interest to know whether the achievable gain scales linearly with $K$ as it goes to infinity, under a cooperation constraint that only limits the number of transmitters at which any message can be available by a cooperation order $M$. We showed that the answer is negative for the fully connected channel where message assignment strategies satisfy the local cooperation constraint, as well as all possible message assignments for the case where $M=2$. The problem is still open for fully connected channels and values of $M \geq 3$.

For locally connected channels where each transmitter is connected to the receiver carrying the same index as well as $L$ neighboring receivers, we showed that the asymptotic per user DoF is lower bounded by $\max\left\{\frac{1}{2},\frac{2M}{2M+L}\right\}$. The achieving coding scheme is simple as it relies only on zero-forcing transmit beamforming. We showed that this lower bound is tight for the case where $L=1$. In particular, the characterized asymptotic per user DoF for that case is $\frac{2M}{2M+1}$, and is higher than previous results in~\cite{Lapidoth-Shamai-Wigger-ISIT07}, and~\cite{Shamai-Wigger-ISIT11}.

We also revealed insights on the optimal way of assigning messages to transmitters under a cooperation order constraint. For instance, we considered a local cooperation constraint, where each message can only be available at a neighborhood of transmitters whose size does not scale linearly with the number of users. While we showed that local cooperation does not achieve a scalable DoF gain for the fully connected channel, we also showed that local cooperation is optimal for locally connected channels. Furthermore, we have shed light on the intimate relation between the selection of message assignments and the design of transmit beams. We have shown that assigning messages to successive transmitters is beneficial for zero-forcing transmit beamforming in locally connected channels as it minimizes the number of receivers at which each message causes undesired interference. However, the same message assignment strategy can be an impediment to other techniques such as asymptotic interference alignment, because the overlap of sets of messages carried by transmit antennas is large for this assignment of messages.

\section*{Acknowledgment}
The authors thank Dr. Mich\`{e}le Wigger and Dr. Shlomo Shamai (Shitz) for interesting discussions related to this paper.

\appendix
\section*{Auxiliary Lemmas for Large Fully Connected Networks Upper Bounds}

\begin{lem}\label{lem:basis}
\begin{equation*}
\text{There exists } i \in [K] \text{ such that } |C_{\{i\}}| \leq M.
\end{equation*}
\end{lem}
\begin{proof}
The statement follows by the pigeonhole principle, since the following holds,
\begin{equation}
\sum_{i=1}^{K} |C_{\{i\}}| = \sum_{i=1}^{K} |{\cal T}_i| \leq MK.
\end{equation} 
\end{proof}

\begin{lem}\label{lem:inductionstep}
For $M\geq 2$, if $\exists {\cal A} \subset [K]$ such that $|{\cal A}|=n < K$, and $|C_{\cal A}| \leq (M-1)n+1$, then $\exists {\cal B} \subseteq [K]$ such that $|{\cal B}|=n+1$, and $|C_{\cal B}| \leq (M-1)(n+1)+1$. 
\end{lem}
\begin{proof}
We only consider the case where $K > (M-1)(n+1)+1$, as otherwise, the statement trivially holds. In this case, we can show that
\begin{equation}\label{eq:inequalityone}
M(K-|C_{\cal A}|) < (K-n)((M-1)(n+1)+2-|C_{\cal A}|).
\end{equation}
The proof of ~\eqref{eq:inequalityone} is available in Lemma~\ref{lem:genineq} below. Note that the left hand side in the above equation is the maximum number of message instances for messages outside the set $C_{\cal A}$, i.e.,
\begin{eqnarray}
\sum_{i\in[K], i \notin {\cal A}} |C_{\{i\}} \backslash C_{\cal A}| &\leq& M(K-|C_{\cal A}|)\nonumber
\\&<&(K-n)((M-1)(n+1)+2-|C_{\cal A}|).\nonumber
\\
\end{eqnarray}
Since the number of transmitters outside the set ${\cal A}$ is $K-n$, it follows by the pigeonhole principle that there exists a transmitter whose index is outside ${\cal A}$ and carries at most $(M-1)(n+1)+1-|C_{\cal A}|$ messages whose indices are outside $C_{\cal A}$. More precisely,
\begin{equation}
\exists i \in [K]\backslash{\cal A}: |C_{\{i\}}\backslash{C_{\cal A}}| \leq (M-1)(n+1)+1-|C_{\cal A}|.
\end{equation}
It follows that there exists a transmitter whose index is outside the set ${\cal A}$ and can be added to the set ${\cal A}$ to form the set ${\cal B}$ that satisfies the statement.  
\end{proof}

\begin{lem}\label{lem:genineq}
If $K \geq (M-1)(n+1)+1$, $M \geq 2$, and $\exists {\cal S}\subseteq [K]$ such that $|{\cal S}| \leq (M-1)n+1$, then the following holds,
\begin{equation}
M(K-|{\cal S}|) < (K-n)\left((M-1)(n+1) + 2 - |{\cal S}|\right).
\end{equation}
\end{lem}
\begin{proof}
We first prove the statement for the case where $|{\cal S}| = (M-1)n+1$. This directly follows as,
\begin{eqnarray}
M(K-|{\cal S}|) &=& M(K-((M-1)n+1))\nonumber
\\&\leq& M(K-(n+1))\nonumber
\\&<& M(K-n)\nonumber
\\&=&(K-n)\left((M-1)(n+1) + 2 - |{\cal S}|\right).\nonumber
\\
\end{eqnarray}
In order to complete the proof, we note that each decrement of $|{\cal S}|$ leads to an increase in the left hand side by $M$, and in the right hand side by $K-n$, and,
\begin{eqnarray}
K-n &\geq&  (M-1)(n+1)+1 - n\nonumber
\\&=&(M-2)n + M\nonumber
\\&\geq& M.
\end{eqnarray} 
\end{proof}

\begin{lem}\label{lem:mthreeinductionstep}
For $M=3$, If $\exists {\cal A} \subset [K]$ such that $|{\cal A}|=n$, and $\frac{K+1}{4} \leq n < K$, $|C_{\cal A}| \leq n + \frac{K+1}{4}+ 1$, then $\exists {\cal B} \subset [K]$ such that $|{\cal B}|=n+1$, $|C_{\cal B}| \leq n + \frac{K+1}{4} + 2$.
\end{lem}

\begin{proof}
The proof follows in a similar fashion to that of Lemma~\ref{lem:inductionstep}. Let $x=n+\frac{K+1}{4} + 1$. We only consider the case where $K > x+1$, as otherwise, the proof is trivial. We first assume the following,
\begin{equation}\label{eq:inequalitytwo}
3(K-|C_{\cal A}|) < (K-n)\left(n+\frac{K+1}{4}+3-|C_{\cal A}|\right).
\end{equation}
Now, it follows that,
\begin{eqnarray}
\sum_{i\in[K], i \notin {\cal A}} |C_{\{i\}} \backslash C_{\cal A}| &\leq& M(K-|C_{\cal A}|)\nonumber
\\&<&(K-n)\left(n+\frac{K+1}{4}+3-|C_{\cal A}|\right),\nonumber
\\
\end{eqnarray}
and hence,
\begin{equation}
\exists i \in [K]\backslash{\cal A}: |C_{\{i\}}\backslash{C_{\cal A}}| \leq n+\frac{K+1}{4}+2-|C_{\cal A}|,
\end{equation}
and then the set ${\cal B}={\cal A}\cup \{i\}$ satisfies the statement of the lemma. Finally, we need to show that~\eqref{eq:inequalitytwo} is true. For the case where $|C_{\cal A}|=x$,
\begin{eqnarray}
3x &=& \frac{3K}{4} + \frac{15}{4} + 3n\nonumber
\\&=& (2n+K) + \left(n - \frac{K}{4} + \frac{15}{4}\right)\nonumber
\\&>& 2n+K,
\end{eqnarray}
and hence, $3(K-x) < 2(K-n)$, which implies~\eqref{eq:inequalitytwo} for the case where $|C_{\cal A}|=x$. Moreover, we note that each decrement of $|C_{\cal A}|$ increases the left hand side of~\eqref{eq:inequalitytwo} by $3$ and the right hand side by $(K-n)$, and we know that,
\begin{eqnarray}
K &>& x+1\nonumber
\\&=& n+\frac{K+1}{4}+2\nonumber
\\&\geq& n+2,
\end{eqnarray}
and hence, $K-n \geq 3$, so there is no loss of generality in assuming that $|C_{\cal A}|=x$ in the proof of~\eqref{eq:inequalitytwo}, and the proof is complete.
\end{proof}

\section*{Proof of Lemma~\ref{lem:dofouterbound}}
In order to prove the lemma, we show that using a reliable communication scheme with the aid of a signal that is within $o(\log P)$, all the messages can be recovered from the set of received signals $Y_{\cal A}$. It follows that any achievable degree of freedom for the channel is also achievable for another channel that has only those receivers, thus proving the upper bound. 

In any reliable $n$-block coding scheme, \[\mathsf{H}(W_i|Y_i^n) \leq n\epsilon, \forall i \in [K].\]
Therefore, \[\mathsf{H}(W_{\cal A}|Y_{\cal A}^n) \leq \sum_{i \in {\cal A}} \mathsf{H}(W_i|Y_i^n) \leq n|{\cal A}|\epsilon.\]
Now, the sum $\sum_{i \in [K]} R_i = \sum_{i \in \bar{{\cal A}}} R_i + \sum_{i \in {\cal A}}R_i$ can be bounded as
\begin{eqnarray}\label{eq:lemma_tmp1} 
n\left(\sum_{i \in \bar{\cal A}} R_i + \sum_{i \in {\cal A}}R_i\right) & = & \mathsf{H}(W_{\bar{\cal A}}) + \mathsf{H}(W_{\cal A})\notag  \\
& \leq & I\left(W_{\bar{\cal A}};Y_{\bar{\cal A}}^n\right) + I\left(W_{\cal A};Y_{\cal A}^n\right)\notag\\&&+nK \epsilon,
\end{eqnarray}
where $\epsilon$ can be made arbitrarily small, by choosing $n$ large enough. The two terms on the right hand side of \eqref{eq:lemma_tmp1} can be bounded as
\begin{eqnarray*}
I\left(W_{\cal A};Y_{\cal A}^n\right) 
& = &  \entropy{Y_{\cal A}^n} - \entropy{Y_{\cal A}^n|W_{\cal A}}\\
& \leq &  \sum_{i \in {\cal A}}\sum_{t = 1}^n \entropy{Y_i(t)} - \entropy{Y_{\cal A}^n|W_{\cal A}} \\
& = &  |{\cal A}|n\log P + n(o(\log P)) - \entropy{Y_{\cal A}^n|W_{\cal A}},
\end{eqnarray*}
\begin{eqnarray*}
I\left(W_{\bar{\cal A}};Y_{\bar{\cal A}}^n\right) & \leq & \ I\left(W_{\bar{\cal A}};Y_{\bar{\cal A}}^n,Y_{\cal A}^n,W_{\cal A}\right) \\
& = &  I(W_{\bar{\cal A}};Y_{\cal A}^n|W_{\cal A}) + I(W_{\bar{\cal A}};Y_{\bar{\cal A}}^n|W_{\cal A},Y_{\cal A}^n) \\
& \leq &  \entropy{Y_{\cal A}^n|W_{\cal A}}  - \entropy{Z_{\cal A}^n} + \entropy{Y_{\bar{\cal A}}^n|W_{\cal A},Y_{\cal A}^n}-\entropy{Z_{\bar{A}}^n}.
\end{eqnarray*}
Now, we have
\begin{eqnarray*}
I\left(W_{\cal A};Y_{\cal A}^n\right)  + I\left(W_{\bar{\cal A}};Y_{\bar{\cal A}}^n\right) &\leq&
  |{\cal A}|n\log P  + \entropy{Y_{\bar{\cal A}}^n|W_{\cal A},Y_{\cal A}^n}\\&&  + n(o(\log P)).
\end{eqnarray*}

Therefore, if we show that \[\entropy{Y_{\bar{\cal A}}^n|W_{\cal A},Y_{\cal A}^n} = n(o(\log P)),\] then from \eqref{eq:lemma_tmp1}, we have the required outer bound. Since $W_{\cal A}$ contains all the messages carried by transmitters with indices $U_{\cal A}$, they determine $X_{U_{\cal A}}^n$. Therefore,
\begin{eqnarray*}
\entropy{Y_{\bar{\cal A}}^n|W_{\cal A},Y_{\cal A}^n}  &=&  \entropy{Y_{\bar{\cal A}}^n|W_{\cal A},Y_{\cal A}^n,X_{U_{\cal A}}^n} \\
&\leq& \entropy{Y_{\bar{\cal A}}^n|Y_{\cal A}^n,X_{U_{\cal A}}^n} \\
& \leq & \sum_{t = 1}^{n} \entropy{Y_{\bar{\cal A}}(t)|Y_{\cal A}(t),X_{U_{\cal A}}(t)} \\
&\overset{(a)}{\leq}& \sum_{t = 1}^{n} \entropy{Y_{\bar{\cal A}}(t)|X_{U_{\cal A}}(t),X_{\bar{U}_{\cal A}}(t)+f_2(Z_{\cal A}(t))}\\
&\overset{(b)}{\leq}& n(o(\log P))
\end{eqnarray*}
where $(a)$ follows from the existence of the function $f_1$ such that $f_1(Y_{\cal A},X_{U_{\cal A}})=X_{\bar{U}_{\cal A}}+f_2(Z_{\cal A})$. Recall that for ${\cal S}_1\subseteq[K],{\cal S}_2\subseteq[K]$, ${\bm H}_{{\cal S}_1,{\cal S}_2}$ denotes the $|{\cal S}_1|\times|{\cal S}_2|$ matrix of channel coefficients between $X_{{\cal S}_2}$ and $Y_{{\cal S}_1}$, then $(b)$ follows as,
\begin{eqnarray*}
Y_{\bar{\cal A}} &=& {\bm H}_{\bar{\cal A},U_{\cal A}}X_{U_{\cal A}}+{\bm H}_{\bar{\cal A},\bar{U}_{\cal A}}X_{\bar{U}_{\cal A}}+Z_{\bar{\cal A}}
\\&=&{\bm H}_{\bar{\cal A},U_{\cal A}}X_{U_{\cal A}}+{\bm H}_{\bar{\cal A},\bar{U}_{\cal A}}\left(X_{\bar{U}_{\cal A}}+f_2(Z_{\cal A})\right)+Z_{\bar{\cal A}}-{\bm H}_{\bar{\cal A},\bar{U}_{\cal A}}f_2(Z_{\cal A}),
\end{eqnarray*}
and hence,
\begin{eqnarray*}
\entropy{Y_{\bar{\cal A}}|X_{U_{\cal A}},X_{\bar{U}_{\cal A}}+f_2(Z_{\cal A})}&\leq&\entropy{Y_{\bar{\cal A}}|{\bm H}_{\bar{\cal A},U_{\cal A}}X_{U_{\cal A}}+{\bm H}_{\bar{\cal A},\bar{U}_{\cal A}}\left(X_{\bar{U}_{\cal A}}+f_2(Z_{\cal A})\right)}
\\&\leq&\entropy{Z_{\bar{\cal A}}-{\bm H}_{\bar{\cal A},\bar{U}_{\cal A}}f_2(Z_{\cal A})}
\\&=&o(\log P)
\end{eqnarray*}

\end{document}